\documentclass{amsart}
\usepackage{amsmath}
\usepackage{amssymb}
\usepackage{tikz}
\usepackage{CJK}
\usepackage{indentfirst}
\usepackage{amsmath,amssymb,amsthm, mathrsfs,latexsym,bm}
\usepackage{setspace}
\usepackage{fancyhdr}
\usepackage{cite}
\usepackage{graphicx}
\usepackage{epigraph}
\usepackage{geometry}
\newtheorem{Example}{Example}[subsection]

\newtheorem{Theorem}{Theorem}[subsection]
\newtheorem{Theorem/Definition}{Theorem/Definition}[subsection]
\newtheorem{Proposition}{Proposition}[subsection]
\newtheorem{Lemma}{Lemma}[subsection]
\newtheorem{Corollary}{Corollary}[subsection]

\theoremstyle{remark}
\newtheorem{Remark}{Remark}[subsection]

\numberwithin{equation}{section}

\newcommand{\half}{\frac{1}{2}}
\newcommand{\bZ}{{\mathbb Z}}  \newcommand{\bt}{{\bf t}}

\newcommand{\pd}{\partial}
\newcommand{\cM}{{\mathcal M}}
\newcommand{\Mbar}{\overline{\cM}}

\newcommand{\be}{\begin{equation}}
\newcommand{\ee}{\end{equation}}
\newcommand{\ben}{\begin{eqnarray*}}
\newcommand{\een}{\end{eqnarray*}}
\newcommand{\bea}{\begin{eqnarray}}
\newcommand{\eea}{\end{eqnarray}}

\usepackage{color}

\definecolor{yellow}{rgb}{1,1,0}
\definecolor{orange}{rgb}{1,.7,0}
\definecolor{red}{rgb}{1,0,0} \definecolor{green}{rgb}{0,1,1}
\definecolor{white}{rgb}{1,1,1}

\definecolor{A}{rgb}{.75,1,.75}

\newcommand{\corr}[1]{\langle {#1} \rangle}

\begin{document}

\title{On Itzykson-Zuber Ansatz}

\author{Qingsheng Zhang}
\address{Department of Mathematical Sciences\\
Tsinghua University\\Beijing, 100084, China}
\email{zqs15@mails.tsinghua.edu.cn}

\author{Jian Zhou}
\address{Department of Mathematical Sciences\\
Tsinghua University\\Beijing, 100084, China}
\email{jianzhou@mail.tsinghua.edu.cn}

\date{}

\pagenumbering{Roman}
\begin{abstract}
We apply the renormalized coupling constants and Virasoro constraints to derive the Itzykson-Zuber Ansatz
on the form of the free energy in topological 2D gravity.
We also treat the topological 1D gravity and the Hermitian one-matrix models in the same fashion.
Some uniform behaviors are discovered in this approach.
\end{abstract}
\maketitle
\setlength{\parindent}{2em}

\epigraph{
\textsf{How can I convert this from a problem of infinite number of degrees of freedom, which you can't deal with anyhow,
to a problem which is finite even if it was so large that you would have to have an astronomical size computer.
I just wanted to convert it from an infinite number of degrees of freedom to a finite number.}}%
{K. Wilson}%


\pagenumbering{arabic}

\section{Introduction}

In recent years there have appeared various Gromov-Witten type theories.
In all these theories we study their partition functions or free energy functions,
expressed as formal power series in infinitely many formal variables.
It is not convenient to study functions in infinitely many variables.
The well-established methods in mathematics usually deal with only finitely many variables,
especially when analyticity or smoothness of the functions are concerned.
It is then very desirable to develop some methods to convert the problems that involves infinitely many
variables to problems with only finitely many variables.
Amazingly,
such situations were faced by Wilson when he studied renormalization theory.
He discovered that in doing renormalzations step by step,
it is necessary to work with Hamiltonians with all possible coupling
constants hence it is necessary to work with a problem of an infinite degrees of  freedom.
By considering the fixed points of the renormalization flow,
a miracle happens: In the limit the theory becomes soluble with finitely many degrees of freedom.
We will report in this work similar miracles happen in the case of some Gromov-Witten type theories.
In this paper we will focus on three examples: 1D topological gravity, Hermitian one-matrix models,
and 2D topological gravity.
In subsequent work,
we will make generalizations to other models.

To explain our results,
let us begin with the case of two-dimensional topological gravity.
Witten \cite{Wit1} interpreted the 2D topological gravity
as the intersection theory of $\psi$-classes on the Deligne-Mumford moduli spaces
$\Mbar_{g,n}$.
The free energy of this theory is  the generating function  defined by:
\be
F^{2D}(t):=\sum_{n_0,n_1,n_2\cdots}
\left<\tau_0^{n_0}\tau_{1}^{n_{1}}\tau_{2}^{n_{2}}\cdots \right>^{2D}
\frac{t_{0}^{n_0}}{n_0!}\frac{t_{1}^{n_1}}{n_1!}\frac{t_{2}^{n_2}}{n_2!}\cdots,
\ee
where $\{t_i\}_{i\geqslant0}$ are formal variables understood as the coupling constants,
and they will be understood as coordinates on the big phase space of the 2D topological gravity.
The partition function is defined by
\be
Z^{2D}:=e^{F^{2D}}.
\ee
Witten conjectured that $Z^{2D}$ is a tau-function of the KdV hierarchy,
i.e., $u= \frac{\pd^2F^{2D}}{\pd t_0^2}$ satisfies a system of nonlinear differential equations of
the form:
\be
\frac{\pd u}{\pd t_n} = \frac{\pd}{\pd t_0} R_{n+1}(u, u_{t_0}, \dots),
\ee
and furthermore, the free energy satisfies the string equation
\be
\frac{\partial F^{2D}}{\partial t_0}=\frac{t_0^2}{2}+\sum_{n=0}^{\infty}t_{n+1}\frac{\partial F^{2D}}{\partial t_n}.
\ee
Conversely, together with the string equation, the KdV hierarchy completely determines the free energy.
This conjecture was proved by Kontsevich \cite{Kon} by computing the partition function
in a different way-the Kontsevich matrix model,
so the partition function $Z^{2D}$ is also known as the Witten-Kontsevich tau-function and denoted by $\tau_{WK}$.

There is another way to study the Witten-Kontsevich tau-function $\tau_{WK}$.
As pointed out by Dijkgraaf-Verlinde-Verlinde\cite{DVV1}, $\tau_{KW}$ can be uniquely
determined by a group of linear partial differential equations known as the Virasoro constraints.
For derivations of the equivalence of the two ways, see \cite{DVV1}.

One can use either the characterization by KdV hierarchy plus string equation or
the Virasoro constraints to compute the free energy functions $F_g$.
The results are of course some formal power series in the coupling constants $t_0, t_1, \dots$
which are too complicated to expect any closed formulas directly.
In \cite{IZ}, Itzykson and Zuber introduced another group of variables $I_0,I_1,\cdots$ defined by:
\begin{align}
u_{0}(t):=&\frac{\partial^2 F_0}{\partial t_{0}^2},\\
I_n(t):=&\sum_{k\geqslant0}t_{n+k}\frac{u_{0}^k}{k!}, \ \ \ n\geqslant 0
\end{align}
and they proved that $u_0=I_0$, i.e.
\be\label{Icoordinates}
I_n(t)=\sum_{k\geqslant0}t_{n+k}\frac{I_{0}^k}{k!}, \ \ \ n\geqslant 0.
\ee
With these definitions, they made an ansatz on the shape of $F_g$:
\be\label{F2D}
F^{2D}_g(t)=\sum_{\sum_{2\leqslant k\leqslant 3g-2}(k-1)l_k=3g-3}\left<\tau_2^{l_2}\tau_{3}^{l_{3}}\cdots\tau_{3g-2}^{l_{3g-2}} \right>^{2D}_g \prod_{j=2}^{3g-2}\frac{1}{l_j!}\left(\frac{I_j}{(1-I_1)^{\frac{2j+1}{3}}}\right)^{l_j}
\ee
for $g\geqslant2$.
This is a finite sum of monomials in $\frac{I_j}{(1-I_1)^{(2j+1)/3}}$.
By inserting this ansatz into KdV hierarchy,
they got some concrete results of free energies of low genus and they remarked that
``There is no difficulty to pursue these computations as far as one wishes".

Inspired by this ansatz especially the introduction of $I_k$,
the second named author carried out a comprehensive study of one-dimensional topological gravity
in \cite{Zhou1} in the hope that this would lead to a better understanding of
the Itsykson-Zuber ansatz.
He considered the the action function:
\be \label{eqn:Action-1D}
S(x)=-\frac{1}{2}x^2+\sum_{n=0}^{\infty}t_{n}\frac{x^{n+1}}{(n+1)!},
\ee
and its renormalization by the iterated procedure of completing the square.
In the limit of this procedure,
one reaches the critical point $x_{\infty}$ of the action function
which satisfies the following equation (\cite{Zhou1}, Proposition 2.1):
\be
x_{\infty}=\sum_{n\geqslant0}t_{n}\frac{x_{\infty}^n}{n!},
\ee
and the Taylor expansion of $S(x)$ at $x=x_{\infty}$ is
\be
S(x)=S(x_{\infty})+\sum_{n=2}^{\infty}(I'_{n-1}-\delta_{n,2})\frac{(x-x_{\infty})^{n}}{n!},
\ee
where
\begin{align}
I'_n(t)=&\sum_{k\geqslant0}t_{n+k}\frac{x_{\infty}^k}{k!}, \ \ \ n\geqslant 1.
\end{align}
One can see that $x_{\infty}$ and $I'_k$ are exactly $I_0$ and $I_k$ by definition!
In \cite{Zhou1}, by Lagrange inversion formula, the second named author derived an explicit formula for $I_0$
in terms of $\{t_n\}_{n \geq 0}$:
\begin{equation}
  I_0=t_0+\sum_{k=1}^{\infty}\frac{1}{k} \sum_{\substack{p_1+\cdots+p_{k}=k-1\\p_i\geqslant0,i=1,\cdots,k}} \frac{t_{p_1}}{p_1!}\cdots\frac{t_{p_{k}}}{p_{k}!}.
\end{equation}
By plugging this into \eqref{Icoordinates},
one can get similar expressions for $I_k$.
He also considered the inverse transformation, and got the following result:
\begin{equation}\label{Itot}
t_n=\sum_{k\geqslant0}I_{n+k}\frac{(-1)^kI_{0}^k}{k!}, \ \ \ n\geqslant 0.
\end{equation}

As we have mentioned,
$\{t_n\}_{n\geqslant0}$ are understood as coordinates on the big phase space.
Now $\{I_n\}_{n\geqslant0}$ can be understood as new coordinates on the big phase space.
They are called the renormalized coupling constants in \cite{Zhou1}.
The reason for this terminology is because the striking analogy with Wilson's renormalization
theory in the sense that one starts with an arbitrary set of infinitely
many coupling constants and the fixed point of the renormalization flow gives one another set
of infinitely coupling constants.
Another analogy with Wilson's theory is that with the introduction of these renormalized
constants makes the theory become soluble in finitely many degrees of freedom.
This is our interpretation of the Itzykson-Zuber ansatz \eqref{F2D}.
In fact, for 1D gravity, inspired by Itzykson-Zuber's ansatz, the second author proved that for the 1D
topological  gravity theory,
\be\label{F1D}
F^{1D}_g(t)=\sum_{\sum_{2\leqslant k\leqslant 2g-1}(k-1)l_k=2g-2}\left<\tau_2^{l_2}\tau_{3}^{l_{3}}\cdots\tau_{2g-1}^{l_{2g-1}} \right>^{1D}_g \prod_{j=2}^{2g-1}\frac{1}{l_j!}\left(\frac{I_j}{(1-I_1)^{\frac{j+1}{2}}}\right)^{l_j}
\ee
for $g\geqslant2$, where $F^{1D}_g$ is the free energy of genus g of the topological 1D gravity.

Another example we treat in this work is the Hermitian one-matrix models of order $N$
(i.e. the matrix is of $N\times N$).
The second author understood  the 1D  topological gravity
as the Hermitian one-matrix models of order $1$
and extended some results for 1D topological gravity
to Hermitian one-matrix models in \cite{Zhou2, Zhou3, Zhou4}.
Using the same method, he proved that for Hermitian one-matrix models,
\be\label{FN}
F^{N}_g(t)=\sum_{\sum_{2\leqslant k\leqslant 2g-1}(k-1)l_k=2g-2}\left<\tau_2^{l_2}\tau_{3}^{l_{3}}\cdots\tau_{2g-1}^{l_{2g-1}} \right>^{N}_g \prod_{j=2}^{2g-1}\frac{1}{l_j!}\left(\frac{I_j}{(1-I_1)^{\frac{j+1}{2}}}\right)^{l_j}
\ee
for $g\geqslant2$, where $F^{N}_g$ is the free energy of genus $g$.

The proofs of \eqref{F1D} and \eqref{FN} in previous works of the second author
are based on rewriting the corresponding puncture equation and the dilaton equation
in 1D topological gravity and Hermitian one-matrix models
in terms of the $I$-coordinates.
It was announced in \cite{Zhou1} that the same method can be applied to establish
the Itzykson-Zuber ansatz.
In this work we will achieve more than that.
More precisely,
we will show that all the Virasoro operators $\{L_{m}\}_{m\geqslant-1}$
in all the above three examples
can be rewritten in terms of the $I$-coordinates,
and using this fact,
one can also calculate the free energies of in arbitrary genus $g$ recursively.

In \cite{Zhou6} the second named author introduced some coupling constants $t_{-n}$ for $n \geq 1$
and call them the ghost variables.
They were used to defined the following extension of free energy $F_0^{2D}$ in genus zero:
\be
\tilde{F}_0^{2D} = F_0^{2D} + \sum_{n \geq 0} (-1)^n (t_n-\delta_{n,1}) t_{-n-1},
\ee
This was used to make sense of the following special deformation of Airy curve
introduced in \cite{Zhou5}:
\be \label{eqn:Airy}
w^{2D}:=z^{\frac{1}{2}}-\sum_{n=0}^{\infty}\frac{t_n}{(2n-1)!!}z^{n-\frac{1}{2}} -\sum_{n=0}^{\infty}(2n+1)!!\frac{\partial F_0^{2D}}{\partial t_n}z^{-n-\frac{3}{2}}.
\ee
See also  \S \ref{sec:Ghost}.
In this paper we also consider the renormalizations of the ghost variables.
They will be denoted by $I_{-k}$ for $k \geq 1$:
\be
I_{-k}=\sum_{n=0}^{\infty}t_{n-k} \frac{I_0^n}{n!}.
\ee
When we impose the condition $t_{-m}=0$ for $m>0$,
the ghost variables $I_{-n}$ can be expressed as formal series in $\{I_{k}\}_{k\geq0}$:
\be
I_{-n}|_{t_{-m}=0, m\geq1}=\sum_{k=0}^{\infty}I_{k}\frac{(-1)^{k}I_{0}^{k+n}}{k!(n-1)!(k+n)}.
\ee
It turns out that certain shift, denoted by $\tilde{I}_{-n}$ is more natural.
An amazing result is that
\bea
&& F_0^{1D} = \tilde{I}_{-1}|_{t_{-m}=0,\ m\geq1}, \\
&& F_0^{N} = N \cdot \tilde{I}_{-1}|_{t_{-m}=0,\ m\geq1}, \\
&& F_0^{2D} = \frac{1}{2}\sum_{n\geqslant1}(-1)^n\tilde{I}_{n}\tilde{I}_{-n-1}|_{t_{-m}=0,\ m\geq1}.
\eea

We also used the renormalized ghost variables to
investigate the applications of the $I$-coordinates to the study of the special deformations
of the emergent spectral curves in the three cases we consider in this paper.
It turns out that they manifest some uniform behaviors related to \eqref{eqn:Action-1D}
in this new perspective.
We summarize them in the end of the paper where we present some concluding remarks.

In summary,
we have studied the Itzykson-Zuber Ansatz and its analogues from the point of view
of Wilson's renormalization theory.
One can interpret the results stated in such Ansatz as generalizations of
the constitutive relations in the mean field theory approach studied by Dijkgraaf and Witten \cite{DW}.
It is interesting to see that renormalization leads to the derivations of results in
mean field theory in these theory.
We believe this should hold in general and hope to return to this in future investigations.

We arrange the rest of the paper in the following fashion.
We treat the cases of topological 1D gravity and Hermitian one-matrix models
in Section \ref{T1D}  and Section \ref{sec:HMM} respectively.
We verify the Itzykson-Zuber Ansatz in Section \ref{sec:IZ}.
We generalize the renormalized coupling constants to include the ghost coupling constants
in Section \ref{sec:Ghost},
and use the renormalized ghost variables to study the special deformation of the Airy curve
induced by the Witten-Kontsevich tau-function.
In Section \ref{sec:MFT} we rederive the constitutive relations in genus zero due to Dijkgraaf and Witten \cite{DW}
and derive their analogues for $F_0^{1D}$ and $F_0^N$.
In the final Section \ref{sec:Remarks}
we comment on the uniform behavior of the special deformations of the spectral curves
in the perspective of $I$-coordinates.

\section{Computations in 1D Topological Gravity
by Virasoro Constraints in Renormalized Coupling Constants}
\label{T1D}

In this section we recall the 1D topological gravity \cite{Zhou1}.
The partition function of this theory satisfies Virasoro constraints derived in \cite{NY} and
further studied in \cite{Zhou1}.
We rewrite the Virasoro constraints in the $I$-coordinates.
Using these constraints, we derive a recursively way to solve free energy in $I$-coordinates.
We also study the special deformation of the 1D gravity in $I$-coordinates.

\subsection{Renormalized coupling constants in the 1D topological gravity}

In order to understand how the $I$-coordinates in the Itzykson-Zuber Ansatz
naturally arise,
the second named author proposed in \cite{Zhou1} to start with
the topological 1D gravity and understand its action function
from the point of view of Wilson's renormalization theory.
The partition function of 1D gravity is the following formal Gaussian integral:
\be
Z^{1D}=\frac{1}{\sqrt{2\pi}\lambda}\int e^{\frac{1}{\lambda^2}S(x)}dx,
\ee
where the action function of the topological 1D gravity given by:
\be \label{def:Action}
S(x)=-\frac{1}{2}x^2+\sum_{n=0}^{\infty}t_{n}\frac{x^{n+1}}{(n+1)!}.
\ee
The coefficients $\{t_n\}_{n\geqslant0}$ are the bare coupling constants of this theory.
One can modify the action by completion of square:
\ben
S  & = & \biggl(t_{-1} +  \half \frac{t_0^2}{1-t_1} + \sum_{n \geq 3}  \frac{t_{n-1}}{n!}
\biggl( \frac{t_0}{1-t_1}\biggr)^{n} \biggr)
+ \tilde{x} \sum_{n \geq 2}  \frac{t_{n}}{n!}  \biggl( \frac{t_0}{1-t_1}\biggr)^{n} \\
& &- \frac{1}{2}\biggl(1
- \sum_{n \geq 0} t_{n+1}   \frac{1}{n!}
\biggl( \frac{t_0}{1-t_1}\biggr)^{n} \biggr)\tilde{x} ^2
+ \sum_{m=3}^\infty \frac{\tilde{x}^m}{m!}
\sum_{n \geq 0} t_{n+m-1} \frac{1}{n!}  \biggl( \frac{t_0}{1-t_1}\biggr)^{n},
\een
where $\tilde{x} = x-x_1$.
A new set of coupling constants is then obtained in this fashion.
As explained in \cite{Zhou1},
by repeating this procedure for infinitely many times,
one gets a limiting set of coupling constants which are fixed by the procedure
of completing the square.
More precisely,
the action function has the following form:
\be\label{eqn:Action}
S(x)=\sum_{k=0}^{\infty}\frac{(-1)^k}{(k+1)!}(I_k+\delta_{k,1})I_{0}^{k+1}+
\sum_{n=2}^{\infty}(I_{n-1}-\delta_{n,2}) \frac{(x-I_0)^{n}}{n!}
\ee
in the limit.
The limiting set of coupling constants $\{I_n\}_{n \geq 0}$
will be referred to as the {\em renormalized coupling constants}.
Here $I_n$ are defined by:
\be \label{def:Ik}
I_k= \sum_{n \geq 0} t_{n+k} \frac{x_\infty^n}{n!},
\ee
where $I_0=x_\infty$ satisfies the equation
$$\frac{\pd S}{\pd x} (x_{\infty}) = 0,$$
or, more explicitly,
\be \label{eqn:Critical}
x_\infty = \sum_{n\geq 0} t_n \frac{x_\infty^n}{n!}.
\ee
This situation is analogous to Wilson's renormalization theory.
First of all,
Wilson considered the space of Hamiltonians with all possible coupling constants.
In our case,
we allow our action function $S$ to have all possible coupling constants $t_n$,
at the expense of considering only formal Gaussian integrals
instead of addressing the issue of convergence of the Gaussian integrals.
Secondly,
Wilson started with a theory with arbitrary coupling constants and modified it to get a new theory
of the same form with different  coupling constants.
In our case,
we use the completion  of square to modify our coupling constants.
Thirdly,
Wilson introduced the notion of a fixed point to describe the limiting theory of
the renormalization flow.
In our case the situation is similar.
We reach the critical point of the action function
and use the expansion there to obtain the limiting coupling constants.
Furthermore,
in \cite{Zhou1} it was proved that the transformation from $\{t_n\}_{n \geq 0}$
to $\{I_n\}_{n\geq 0}$ can be regarded as a nonlinear change of coordinates,
and the space of the theory of 1D topological gravity can be regarded as an infinite-dimensional
manifold with at least two coordinate patches given by local coordinates $\{t_n\}_{n \geq 0}$
and $\{I_n\}_{n \geq 0}$ respectively.

It turns out that the analogue of the Itzykson-Zuber Ansatz also holds in 1D topological gravity.
The free energy of 1D topological gravity is defined by:
\be
F^{1D}=\log Z^{1D}.
\ee
There is a genus expansion for $F^{1D}$:
\be\label{genusexpansion}
F^{1D}=\sum_{g=0}^{\infty}\lambda^{2g-2}F^{1D}_g,
\ee
where $\{F_g^{1D}\}_{g\geqslant0}$ are formal power series of $t_0,t_1, \cdots$. By \cite{Zhou1}, if we define
\be
\deg{t_n}=n-1, \ \ \ n=0,1,2,\cdots
\ee
then $F^{1D}_{g}$ is weighted homogeneous in $t_0,t_1, \cdots$ with
\be
\deg{F_{g}^{1D}}=2g-2,\ \ \ g=0,1,2,\cdots
\ee
By \eqref{def:Ik},
it is natural to define
\be
\deg{I_n}=n-1,\ \ \ n=0,1,2,\cdots
\ee
and then $F_{g}^{1D}$ can be viewed as weighted homogeneous formal series of degree $2g-2$ in $I$-coordinates.
In \cite{Zhou1}, the second named author used two different methods, the Feynman diagram technique and the Virasoro constraints, to get the following results:
\begin{Theorem} \label{thm:F1D}
(\cite[Theorem 6.4 and 6.6]{Zhou1})
\begin{align}
F^{1D}_0=&\sum_{k=0}^{\infty}\frac{(-1)^k}{(k+1)!}(I_k+\delta_{k,1})I_{0}^{k+1},\label{eqn:F01D}\\
F^{1D}_1=&\frac{1}{2}\log{\frac{1}{1-I_1}},\label{eqn:F11D}\\
F^{1D}_g=&\sum_{\sum_{2\leqslant k\leqslant 2g-1}m_k(k-1)=2g-2}\left<\prod_{j=2}^{2g-1}\tau_{j}^{m_j}\right>^{1D}_g\cdot \prod_{j=2}^{2g-1}\frac{1}{m_j!}\left(\frac{I_{j}}{(1-I_1)^{\frac{j+1}{2}}}\right)^{m_j} ,\ g\geqslant2 , \label{eqn:Fg1D}
\end{align}
where the correlators are defined by:
\be
\left<\tau_{a_1}\cdots\tau_{a_n}\right>^{1D}_g=\frac{\partial^n F^{1D}_g}{\partial t_{a_1}\cdots \partial t_{a_n}}\bigg|_{t=0}.
\ee
\end{Theorem}

\subsection{Virasoro constraints for 1D topological gravity in $I$-coordinates}

Let us now recall how to prove the above Theorem by  Virasoro constraints.
First one has the following Theorem:

\begin{Theorem}(Virasoro constraints\cite{Zhou1})
The partition function $Z^{1D}$ of topological 1D gravity satisfies the following equations for $m\geqslant-1$:
\be
L^{1D}_mZ^{1D}=0,
\ee
where
\begin{align}
L^{1D}_{-1}&=\frac{t_0}{\lambda^2}+\sum_{n\geqslant 1}(t_n-\delta_{n,1})\frac{\partial}{\partial t_{n-1}},\\
L^{1D}_{0}&=1+\sum_{n\geqslant 0}(n+1)(t_n-\delta_{n,1})\frac{\partial}{\partial t_{n}},\\
L^{1D}_{m}&=\lambda^2(m+1)!\frac{\partial}{\partial t_{m-1}}+\sum_{n\geqslant 0}\frac{(m+n+1)!}{n!}(t_n-\delta_{n,1})\frac{\partial}{\partial t_{m+n}},\label{eqn:Virasoro1D}
\end{align}
for $m\geqslant1$. Furthermore, $\{L^{1D}_m\}_{m\geqslant-1}$ satisfies the following commutation relations:
\be
\left[L^{1D}_m,L^{1D}_n\right]=(m-n)L^{1D}_{m+n},
\ee
for $m,n\geqslant-1$.
\end{Theorem}

Next the change of coordinates between $\{t_n\}_{n \geq 0}$
and $\{I_n\}_{n\geq 0}$ induces the linear transformations between
$\{\frac{\pd}{\pd t_n}\}_{n \geq 0}$
and $\{\frac{\pd}{\pd I_n}\}_{n\geq 0}$.
The concrete expressions are give by:

\begin{Lemma}(\cite{Zhou1}, Corollary 2.7 and Proposition 2.8)
The vector fields $\{\frac{\partial }{\partial I_k}\}_{k\geqslant0}$  can be expressed in terms of the vector fields $\{\frac{\partial }{\partial t_k}\}_{k\geqslant0}$ as follows:
\begin{align}
\frac{\partial}{\partial I_0}&=\frac{\partial}{\partial t_0}-\sum_{k\geqslant0}t_{k+1}\frac{\partial}{\partial t_k},\label{string} \\
\frac{\partial}{\partial I_l}&=\sum_{k=0}^{l}\frac{(-1)^{l-k}I_0^{l-k}}{(l-k)!}\frac{\partial}{\partial t_k},\ \ \ l\geqslant1\label{vttoI}
\end{align}
and conversely,
\begin{align}
\frac{\partial}{\partial t_k}=\frac{1}{1-I_1}\frac{I_0^{k}}{k!}\bigg(\frac{\pd}{\pd I_0}
+\sum_{l\geqslant1}I_{l+1}\frac{\pd}{\pd I_l}\bigg) +\sum_{1\leqslant l \leqslant k}\frac{I_0^{k-l}}{(k-l)!}\frac{\partial}{\partial I_l},\ \ \ k\geqslant0 \label{vItot}
\end{align}
\end{Lemma}

Using these expressions,
the following formula for $L^{1D}_{-1}$ and $L^{1D}_0$ has been proved \cite[(192), (196)]{Zhou1}:
\begin{align}
L^{1D}_{-1}=&-\frac{\partial}{\partial I_0}+\frac{1}{\lambda^2}
\sum_{n=0}^{\infty}\frac{(-1)^nI_0^n}{n!}I_n, \label{Istring}\\
L^{1D}_0=&-I_0\frac{\partial}{\partial I_0}-2\frac{\partial}{\partial I_1}
+\sum_{l\geqslant 1}(l+1)I_l\frac{\partial}{\partial I_l}+1. \label{Idilaton}
\end{align}
From these one can derive \eqref{eqn:F01D}-\eqref{eqn:Fg1D}.
The main result of this Section is that
by expressing $L^{1D}_m$ for higher $m$ in $I$-coordinates,
one can find an algorithm to explicitly compute $F^{1D}_g$.
In this Subsection we first express $L^{1D}_m$ in $I$-coordinates.
The applications to the computations of $F^{1D}_g$ will be presented
in the next Subsection.

\begin{Theorem}\label{1Dcase}
In the $I$-coordinates,
the Virasoro operators $L_m$ ($m \geq 1$) for 1D topological gravity can be rewritten
  as follows:
\begin{align}\label{IVirasoro}
L^{1D}_{m}=&\lambda^2(m+1)!\bigg(\frac{1}{1-I_1}\frac{I_0^{m-1}}{(m-1)!}\big(\frac{\pd}{\pd I_0}
+\sum_{l\geqslant1}I_{l+1}\frac{\pd}{\pd I_l}\big)
+\sum_{1\leqslant l \leqslant m-1}\frac{I_0^{m-1-l}}{(m-1-l)!}\frac{\partial}{\partial I_l}\bigg)\\
&-I_{0}^{m+1}\frac{\partial}{\partial I_0}
+\sum_{i=1}^{m+1}\binom{m+1}{i}I_0^{m+1-i}
\sum_{p=1}^{\infty}\frac{(p+i)!}{p!}(I_{p}-\delta_{p,1})\frac{\partial}{\partial I_{p+i-1}}.\nonumber
\end{align}
\end{Theorem}
\begin{proof}
The first term in the definition of $L^{1D}_m$ in \eqref{eqn:Virasoro1D} is:
\begin{align}
\frac{\partial}{\partial t_{m-1}}=\frac{1}{1-I_1}\frac{I_0^{m-1}}{(m-1)!}\big(\frac{\pd}{\pd I_0}
+\sum_{l\geqslant1}I_{l+1}\frac{\pd}{\pd I_l}\big)
+\sum_{1\leqslant l \leqslant m-1}\frac{I_0^{m-1-l}}{(m-1-l)!}\frac{\partial}{\partial I_l}.
\end{align}
This is the first term of right hand of \eqref{IVirasoro}.
By \eqref{Itot} and \eqref{vItot},
\begin{align}\label{1}
&\sum_{n\geqslant 0}\frac{(m+n+1)!}{n!}t_n\frac{\partial}{\partial t_{m+n}}\\
=&\sum_{n\geqslant 0}\frac{(m+n+1)!}{n!}\sum_{k\geqslant0}
\frac{(-1)^kI_0^k}{k!}I_{n+k} \frac{I_0^{m+n}}{(m+n)!}\frac{1}{1-I_1}
\left(\frac{\partial}{\partial I_0}
+\sum_{l\geqslant1}I_{l+1}\frac{\partial}{\partial I_l}\right)\nonumber\\
&+\sum_{n\geqslant 0}\frac{(m+n+1)!}{n!}\sum_{k=0}^{\infty}
\frac{(-1)^kI_0^k}{k!}I_{n+k} \left(\sum_{1\leqslant l \leqslant m+n}
\frac{I_0^{m+n-l}}{(m+n-l)!}\frac{\partial}{\partial I_l}\right)\nonumber\\
=&\sum_{p\geqslant 0}\sum_{n+k=p}\frac{m+n+1}{n!k!}(-1)^kI_0^{m+n+k}I_{n+k} \frac{1}{1-I_1}\left(\frac{\partial}{\partial I_0}+\sum_{l\geqslant1}I_{l+1}\frac{\partial}{\partial I_l}\right)\nonumber\\
&+\sum_{l\geqslant1}\sum_{p\geqslant 0}\sum_{n+k=p}\frac{(m+n+1)!}{n!k!(m+n-l)!}(-1)^kI_0^{m+n+k-l}I_{n+k} \frac{\partial}{\partial I_l}\nonumber\\
=&\frac{(m+1+I_1)I_{0}^{m+1}}{1-I_1}\left(\frac{\partial}{\partial I_0}+\sum_{l\geqslant1}I_{l+1}\frac{\partial}{\partial I_l}\right) +\sum_{l\geqslant1,p\geqslant0}\sum_{i=0}^{l+1}\binom{m+1}{i}\binom{l+1}{i}i!\delta_{p,l+1-i} I_0^{m+p-l}I_p\frac{\partial}{\partial I_l}\nonumber\\
=&\frac{(m+1+I_{1})I_{0}^{m+1}}{1-I_1}\left(\frac{\partial}{\partial I_0}+\sum_{l\geqslant1}I_{l+1}\frac{\partial}{\partial I_l}\right) +\sum_{l\geqslant1}\sum_{i=0}^{l+1}\binom{m+1}{i}\binom{l+1}{i}i!I_0^{m+1-i}I_{l+1-i}\frac{\partial}{\partial I_l}. \nonumber
\end{align}
We have used the following identity:
\begin{equation}
\sum_{n+k=p}\frac{(m+n+1)!}{n!k!(m+n-l)!}(-1)^k=\sum_{i=0}^{l+1}\binom{m+1}{i}\binom{l+1}{i}i!\delta_{p,l+1-i}.
\end{equation}
This can be proved as follows:
\begin{align}
\sum_{n+k=p}\frac{(m+n+1)!}{n!k!(m+n-l)!}(-1)^k\nonumber
=&\sum_{n+k=p}\binom{m+n+1}{l+1}\frac{(l+1)!}{n!k!}(-1)^k\nonumber\\
=&\sum_{n+k=p}\sum_{i+j=l+1}\binom{m+1}{i}\binom{n}{j}\frac{(l+1)!}{n!k!}(-1)^k\nonumber\\
=&\sum_{i+j=l+1}\binom{m+1}{i}\frac{(l+1)!}{j!}\sum_{n+k=p}\frac{(-1)^k}{(n-j)!k!}\nonumber\\
=&\sum_{i=0}^{l+1}\binom{m+1}{i}\binom{l+1}{i}i!\delta_{p,l+1-i}\nonumber.
\end{align}
By \eqref{vItot}:
\begin{equation}
\frac{\partial}{\partial t_{m+1}}=\frac{1}{1-I_1}\frac{I_0^{m+1}}{(m+1)!}\frac{\partial}{\partial I_0}+\frac{I_0^{m+1}}{(m+1)!}\sum_{l\geqslant1}\frac{I_{l+1}}{1-I_1}\frac{\partial}{\partial I_l} +\sum_{1\leqslant l \leqslant m+1}\frac{I_0^{m+1-l}}{(m+1-l)!}\frac{\partial}{\partial I_l}.
\end{equation}
Therefore
\begin{align*}
&\sum_{n\geqslant 0}\frac{(m+n+1)!}{n!}(t_n-\delta_{n,1})\frac{\partial}{\partial t_{m+n}}\nonumber\\
=&\frac{(m+1+I_1)I_{0}^{m+1}}{1-I_1}\left(\frac{\partial}{\partial I_0}+\sum_{l\geqslant1}I_{l+1}\frac{\partial}{\partial I_l}\right) +\sum_{l\geqslant1}\sum_{i=0}^{l+1}\binom{m+1}{i}\frac{(l+1)!}{(l+1-i)!}I_0^{m+1-i}I_{l+1-i}
\frac{\partial}{\partial I_l} \nonumber\\
&-\frac{(m+2)I_{0}^{m+1}}{1-I_1}\left(\frac{\partial}{\partial I_0}+\sum_{l\geqslant1}I_{l+1}\frac{\partial}{\partial I_l}\right) -\sum_{1\leqslant l \leqslant m+1}\frac{(m+2)!}{(m+1-l)!}I_0^{m+1-l}\frac{\partial}{\partial I_l}\nonumber\\
=&-I_{0}^{m+1}\left(\frac{\partial}{\partial I_0}+\sum_{l\geqslant1}I_{l+1}\frac{\partial}{\partial I_l}\right)  +\sum_{l\geqslant1}\sum_{i=1}^{l}\binom{m+1}{i}\frac{(l+1)!}{(l+1-i)!}I_0^{m+1-i}I_{l+1-i}
\frac{\partial}{\partial I_l} \nonumber\\
& +\sum_{l\geqslant1}I_{0}^{m+1}I_{l+1}\frac{\partial}{\partial I_l} +\sum_{l\geqslant1}\binom{m+1}{l+1}\frac{(l+1)!}{0!}I_0^{m+1-l}\frac{\partial}{\partial I_l} -\sum_{1\leqslant l \leqslant m+1}\frac{(m+2)!}{(m+1-l)!}I_0^{m+1-l}\frac{\partial}{\partial I_l}\nonumber\\
=&-I_{0}^{m+1}\frac{\partial}{\partial I_0} +\sum_{i=1}^{m+1}\binom{m+1}{i}I_0^{m+1-i} \sum_{l\geqslant i}\frac{(l+1)!}{(l+1-i)!}I_{l+1-i}\frac{\partial}{\partial I_l}  -\sum_{l=1}^{m+1}\frac{(m+1)!}{(m+1-l)!}(l+1)I_0^{m+1-l}\frac{\partial}{\partial I_l}\nonumber\\
=&-I_{0}^{m+1}\frac{\partial}{\partial I_0} +\sum_{i=1}^{m+1}\binom{m+1}{i}I_0^{m+1-i} \sum_{p=1}^{\infty}\frac{(p+i)!}{p!}I_{p}\frac{\partial}{\partial I_{p+i-1}} -\sum_{l=1}^{m+1}\binom{m+1}{l}(l+1)!I_0^{m+1-l}\frac{\partial}{\partial I_l}\nonumber\\
=&-I_{0}^{m+1}\frac{\partial}{\partial I_0} +\sum_{i=1}^{m+1}\binom{m+1}{i}I_0^{m+1-i} \sum_{p=1}^{\infty}\frac{(p+i)!}{p!}(I_{p}-\delta_{p,1})\frac{\partial}{\partial I_{p+i-1}}.
\end{align*}

by \eqref{vItot},
so we have completed the proof.
\end{proof}

\subsection{Computations of $F^{1D}_g$ by Virasoro constraints in $I$-coordinates}

As applications of the expressions of the Virasoro operators expressed in $I$-coordinates
derived in last Subsection,
we use them in this Subsection to compute $F^{1D}_g$.
The formulas for $F^{1D}_0$ and $F^{1D}_1$ in the $I$-coordinate are already known.
We will focus on $F^{1D}_g(t)$ for $g\geqslant2$.

\begin{Theorem}\label{1Dgravity}
For free energy of 1D gravity of genus  $g\geqslant2$, the following equations hold:
\begin{align}
\frac{\partial F^{1D}_{g}}{\partial I_1}=&
\frac{1}{2(1-I_1)}\sum_{l=2}^{2g-1}(l+1)I_l\frac{\partial F^{1D}_{g}}{\partial I_l}, \label{a}\\
\frac{\partial F^{1D}_{g}}{\partial I_{2}}
=& \frac{1}{3(1-I_1)}\left(\frac{1}{1-I_1}\sum_{l=1}^{2g-1}I_{l+1}\frac{\partial F^{1D}_{g-1}}{\partial I_l}
+\sum_{p=2}^{2g-2}\binom{p+2}{p}I_p\frac{\partial F^{1D}_{g}}{\partial I_{p+1}}\right),\label{b}\\
\frac{\partial F^{1D}_{g}}{\partial I_{k+2}}=&\frac{1}{(k+3)(1-I_1)}
\left(\frac{\partial F^{1D}_{g-1}}{\partial I_k}
+\sum_{p=2}^{2g-2-k}\binom{p+k+2}{p}I_p\frac{\partial F^{1D}_{g}}{\partial I_{k+p+1}}\right),\label{c}
\end{align}
where $k=1,2,\cdots,2g-3$.
\end{Theorem}

\begin{proof}

The Virasoro constraints for partition function
$$L^{1D}_mZ^{1D}=0, \ \ \ m\geqslant0$$
can be rewritten as the following equations for free energy in I-coordinate:
\begin{align}
0&=-I_0\frac{\partial F^{1D}}{\partial I_0}-2\frac{\partial F^{1D}}{\partial I_1}+\sum_{l\geqslant 1}(l+1)I_l\frac{\partial F^{1D}}{\partial I_l}+1,\label{L0}\\
0&=\lambda^2(m+1)m\frac{I_0^{m-1}}{1-I_1}\left(\frac{\partial F^{1D}}{\partial I_0}+\sum_{l\geqslant1}I_{l+1}\frac{\partial F^{1D}}{\partial I_l} +(1-I_1)\sum_{1\leqslant l \leqslant m-1}\binom{m-1}{l}l!I_0^{-l}\frac{\partial F^{1D}}{\partial I_l}\right)\label{Lm}\\
&\ \ \ -I_0^{m+1}\frac{\partial F^{1D}}{\partial I_0}+\sum_{i=1}^{m+1}\binom{m+1}{i}\sum_{p=1}^{\infty} \frac{(i+p)!}{p!}(I_{p}-\delta_{p,1})I_0^{m-i+1}\frac{\partial F^{1D}}{\partial I_{i+p-1}}.\nonumber	
\end{align}
Recall the genus expansion of the free energy:
$$F^{1D}=\sum_{g=0}^{\infty}\lambda^{2g-2}F^{1D}_g,$$
we have by \eqref{L0}:
\begin{align}
I_0\frac{\partial F^{1D}_0}{\partial I_0}+2\frac{\partial F^{1D}_0}{\partial I_1}&=\sum_{l\geqslant 1}(l+1)I_l\frac{\partial F^{1D}_0}{\partial I_l}\\
2\frac{\partial F^{1D}_1}{\partial I_1}&=\sum_{l\geqslant 1}(l+1)I_l\frac{\partial F^{1D}_1}{\partial I_l}+1 \label{F1D1}\\
2\frac{\partial F^{1D}_g}{\partial I_1}&=\sum_{l\geqslant 1}(l+1)I_l\frac{\partial F^{1D}_g}{\partial I_l},\ \ \ g>1
\end{align}
and by \eqref{Lm}:
\begin{align}
0&=(m+1)m\frac{I_0^{m-1}}{1-I_1}\left(\sum_{l\geqslant1}I_{l+1}\frac{\partial F^{1D}_{g-1}}{\partial I_l} +(1-I_1)\sum_{1\leqslant l \leqslant m-1}\binom{m-1}{l}l!I_0^{-l}\frac{\partial F^{1D}_{g-1}}{\partial I_l}\right)\label{Fg1D}\\
&\ \ \ +\sum_{i=1}^{m+1}\binom{m+1}{i}\sum_{p=1}^{\infty} \frac{(i+p)!}{p!}(I_{p}-\delta_{p,1})I_0^{m-i+1}\frac{\partial F^{1D}_{g}}{\partial I_{i+p-1}},\nonumber
\end{align}
for $m\geqslant1,\ g\geqslant 2$.
Comparing the coefficients of $I_0^k$ in \eqref{Fg1D} for $k=m+1,m-1,m-2,\cdots,0$, we have:
\begin{align}
0=&\sum_{p\geqslant 1}(p+1)(I_p-\delta_{p,1})\frac{\partial F^{1D}_g}{\partial I_p}\\
0=&\sum_{l\geqslant1}\frac{I_{l+1}}{1-I_1}\frac{\partial F^{1D}_{g-1}}{\partial I_l}+\sum_{p\geqslant 1}\binom{p+2}{2}(I_p-\delta_{p,1})\frac{\partial F^{1D}_g}{\partial I_{p+1}}\\
0=&\frac{\partial F^{1D}_{g-1}}{\partial I_l}+\sum_{p\geqslant 1}\binom{p+l+2}{l+2}(I_p-\delta_{p,1})\frac{\partial F^{1D}_g}{\partial I_{p+l+1}}.
\end{align}
By \eqref{eqn:Fg1D}, $F_g^{1D}$ depends only on $I_1,I_2,\cdots,I_{2g-1}$, this completes the proof.
\end{proof}

\begin{Remark}Theorem \ref{1Dgravity} can be derived in another way: one can firstly rewrite $L_{-1}$ in I-coordinates, this will give $F^{1D}_0$ and show that $F_{g}$ is independent of $I_0$ for $g>0$. To solve $F_g$ for $g>0$, one can let $I_0=0$, in this case, $t_0=0$,
\be
 t_k=I_k, \ \ \ k>0
\ee
and
\be
\begin{split}
\frac{\partial}{\partial t_0}&=\frac{1}{1-I_1}\frac{\partial}{\partial I_0}+\sum_{l\geqslant1}\frac{I_{l+1}}{1-I_1}\frac{\partial}{\partial I_l},\\
\frac{\partial}{\partial t_k}&=\frac{\partial}{\partial I_k}, \ \ \ k>0
\end{split}
\ee
hence one can rewrite $\{L_m\}_{m\geqslant0}$ in I-coordinates in a simpler way.
\end{Remark}

Now we explain how to use Theorem \ref{1Dgravity} to calculate free energies of higher genus.
By \eqref{eqn:Fg1D}, $F^{1D}_g$ depends only on $I_1, \dots, I_{2g-1}$ for $g\geqslant 2$.
By \eqref{c}, we have:
\begin{align}
\frac{\partial F^{1D}_{g}}{\partial I_{2g-1}}=&\frac{1}{2g(1-I_1)}\left(\frac{\partial F^{1D}_{g-1}}{\partial I_{2g-3}}\right),\label{eqn:Fg1Dstart} \\
\frac{\partial F^{1D}_{g}}{\partial I_{2g-2}}=&\frac{1}{(2g-1)(1-I_1)}\left(\frac{\partial F^{1D}_{g-1}}{\partial I_{2g-4}}+\binom{2g}{2}I_2\frac{\partial F^{1D}_{g}}{\partial I_{2g-1}}\right),\\
\vdots & \\
\frac{\partial F^{1D}_{g}}{\partial I_{3}}=&\frac{1}{4(1-I_1)}\left(\frac{\partial F^{1D}_{g-1}}{\partial I_1}+\sum_{p=1}^{2g-3}\binom{p+3}{p}I_p\frac{\partial F^{1D}_{g}}{\partial I_{p+2}}\right).
\end{align}
For $\frac{\partial F^{1D}_{g}}{\partial I_{2}}$
and $\frac{\partial F^{1D}_{g}}{\partial I_{1}}$
we have by \eqref{a} and \eqref{b}:
\begin{align}
\frac{\partial F^{1D}_{g}}{\partial I_{2}}=& \frac{1}{3(1-I_1)}\left(\frac{1}{1-I_1}\sum_{l=1}^{2g-1}I_{l+1}\frac{\partial F^{1D}_{g-1}}{\partial I_l}+\sum_{p=2}^{2g-2}\binom{p+2}{p}I_p\frac{\partial F^{1D}_{g}}{\partial I_{p+1}}\right),\\
\frac{\partial F^{1D}_{g}}{\partial I_1}=&\frac{1}{2(1-I_1)}\sum_{l=2}^{2g-1}(l+1)I_l\frac{\partial F^{1D}_{g}}{\partial I_l}.\label{eqn:Fg1Dend}
\end{align}
By equations \eqref{eqn:Fg1Dstart}-\eqref{eqn:Fg1Dend},
we can solve $\left\{\frac{\pd F^{1D}_{g}}{\partial I_k}\right\}_{g\geqslant2,k=2g-1,2g-2,\cdots,1}$ recursively, given the computation for $F^{1D}_{g-1}$.
Since $F_{g}^{1D}$ is weighted homogeneous of degree $2g-2$ in $I_k$ with $\deg I_k = k-1$,
we have
\be
F^{1D}_g=\frac{1}{2g-2}\sum_{k=1}^{2g-1}(k-1)I_k\frac{\partial F^{1D}_{g}}{\partial I_k}.
\ee

\begin{Example}
Let us compute $F^{1D}_{2}$ by the above procedure.
We have
\begin{align}
\frac{\partial F^{1D}_{2}}{\partial I_{3}}=&\frac{1}{4(1-I_1)}\left(\frac{\partial F^{1D}_1}{\partial I_{1}}\right) =\frac{1}{8(1-I_1)^2},\\
\frac{\partial F^{1D}_{2}}{\partial I_{2}}=&\frac{1}{3(1-I_1)}\left(\frac{1}{1-I_1}I_{2}\frac{\partial F^{1D}_1}{\partial I_1}+\binom{4}{2}I_2\frac{\partial F^{1D}_{2}}{\partial I_{3}}\right) =\frac{5}{12}\frac{I_2}{(1-I_1)^3},\\
\frac{\partial F^{1D}_{2}}{\partial I_{1}}=&\frac{1}{2(1-I_1)}\left(3I_2\frac{\partial F^{1D}_{2}}{\partial I_2}+4I_3\frac{\partial F^{1D}_{2}}{\partial I_3}\right)=\frac{5}{8}\frac{I_2^2}{(1-I_1)^4}+\frac{I_3}{4(1-I_1)^{3}},
\end{align}
therefore
\begin{equation}
F^{1D}_2=\frac{1}{2}\left(I_2\frac{\partial F^{1D}_{2}}{\partial I_2}+2I_3\frac{\partial F^{1D}_{2}}{\partial I_3}\right) =\frac{5}{24}\frac{I_2^2}{(1-I_1)^3}+\frac{1}{8}\frac{I_3}{(1-I_1)^{2}}.
\end{equation}
Similarly, we have with the help of a computer program:
\begin{align}
F^{1D}_3=&\frac{15}{48}\frac{I_2^4}{(1-I_1)^6}+\frac{25}{48}\frac{I_2^2I_3}{(1-I_1)^{6}} +\frac{1}{12}\frac{I_3^2}{(1-I_1)^4} +\frac{7}{48}\frac{I_2I_4}{(1-I_1)^4}+\frac{1}{48}\frac{I_5}{(1-I_1)^3},\\
F^{1D}_4=&\frac{1105}{1152}\frac{I_2^6}{(1-I_1)^9}+\frac{985}{384}\frac{I_2^4I_3}{(1-I_1)^{8}} +\frac{445}{288}\frac{I_2^2I_3^2}{(1-I_1)^7}+\frac{11}{96}\frac{I_3^3}{(1-I_1)^6} +\frac{161}{192}\frac{I_2^3I_4}{(1-I_1)^7}+\frac{7}{12}\frac{I_2I_3I_4}{(1-I_1)^6}\\
&+\frac{21}{640}\frac{I_4^2}{(1-I_1)^{5}} +\frac{113}{576}\frac{I_2^2I_5}{(1-I_1)^6}+\frac{5}{96}\frac{I_3I_5}{(1-I_1)^5} +\frac{1}{32}\frac{I_2I_6}{(1-I_1)^5} +\frac{1}{384}\frac{I_7}{(1-I_1)^4}. \nonumber
\end{align}

\end{Example}

\subsection{Special deformation of the spectral curve of 1D topological gravity in $I$-coordinates}
In \cite{Zhou1}, the second named author defined the special deformation of 1D topological gravity by:
\be
y^{1D}=\frac{1}{\sqrt{2}}\sum_{n\geqslant0}\frac{t_n-\delta_{n,1}}{n!}z^n+\frac{\sqrt{2}}{z} +\sqrt{2}\sum_{n\geqslant1}\frac{n!}{z^{n+1}}\frac{\partial F_0^{1D}}{\partial t_{n-1}}.
\ee
This is a deformation of the Catalan curve \cite{Zhou4}.
Now we rewrite it in $I$-coordinates.
We get the following result:
\begin{Theorem}\label{thm:specialdeformation1D}
In the $I$-coordinates, the special deformation of the spectral curve of 1D
topological gravity can be written as:
\begin{align} \label{eqn:y1D-I}
y^{1D}=&\frac{\sqrt{2}}{z-I_0}+\frac{1}{\sqrt{2}}\sum_{n\geqslant1}\frac{I_n-\delta_{n,1}}{n!}(z-I_0)^n.
\end{align}
\end{Theorem}

\begin{proof}
This is just the $N=1$ case of \cite[Theorem 2.1]{Zhou4}.
\end{proof}

\section{Computations in Hermitian One-Matrix Models by the Renormalized Coupling Constants}
\label{sec:HMM}

In this section we recall the results on Hermitian one-matrix models in \cite{Zhou2, Zhou3, Zhou4}.
The partition function of this theory satisfies Virasoro constraints.
Similar to the case of topological 1D gravity,
we rewrite the Virasoro constraints for Hermitian one-matrix models in $I$-coordinates
and use them to derive the explicit formulas for the free energy in $I$-coordinates.

\subsection{Free energy functions of the Hermitian one-matrix models}

For standard references on matrix models,
see e.g. \cite{Mehta, DGZ}.
Here we follow the notations in \cite{Zhou2, Zhou3, Zhou4}.
The partition function of the Hermitian $N\times N$-matrix model is defined by the formal Gaussian integral:
\be
Z^{N}=\frac{\int_{\mathbb{H}_N}dM \exp{\left(\frac{1}{g_s}tr\left(-\frac{1}{2}M^2+\sum_{n=0}^{\infty}t_n\frac{M^{n+1}}{(n+1)!}\right)\right)}} {\int_{\mathbb{H}_N}dM \exp{\left(-\frac{1}{2g_s}trM^2\right)}},
\ee
where $\mathbb{H}_N$ is the space of Hermitian $N\times N$-matrices.
One can see that for $N=1$,
\be
Z^{N=1}=Z^{1D}.
\ee
The following result is well known, for the proof, one can see \cite{Mehta, DGZ}.
\begin{Proposition}\label{PropMM}
For the Hermitian one-matrix integrals, one has:
\be
\int_{\mathbb{H}_N}dM \exp{\left(\frac{1}{g_s}trV(M)\right)}
=\int_{\mathbb{R}^{N}}\prod_{1\leqslant i<j\leqslant N}(\lambda_i-\lambda_{j})^2
\exp{\left(\frac{1}{g_s}\sum_{i=1}^{N} V(\lambda_i)\right)}\prod_{i=1}^{N}d\lambda_i,
\ee
where $\lambda_1,\lambda_2,\cdots,\lambda_N$ are eigenvalues of $M$.
\end{Proposition}
By taking $V(M)=-\frac{1}{2}M^2+\sum_{n=0}^{\infty}t_n\frac{M^{n+1}}{(n+1)!}$ in this Proposition,
we get the following analogue of the renormalization of topological 1D gravity:
\begin{Proposition}\label{Prop:RenormalHMM}
For the Hermitian one-matrix integrals, one has:
\begin{align}
&\int_{\mathbb{H}_N}dM\exp{\left(\frac{1}{g_s}tr\left(-\frac{1}{2}M^2
+\sum_{n=0}^{\infty}t_n\frac{M^{n+1}}{(n+1)!}\right)\right)}\\
=&\exp\left(\frac{N}{g_s}\sum_{k=0}^{\infty}\frac{(-1)^k}{(k+1)!}(I_k+\delta_{k,1})I_{0}^{k+1}\right)\cdot
\int_{\mathbb{H}_N}dM\exp\left(\frac{1}{g_s}tr
\left(-\frac{1}{2}M^2+\sum_{n=1}^{\infty}I_n\frac{M^{n+1}}{(n+1)!}\right)\right).\nonumber
\end{align}
\end{Proposition}
\begin{proof}
By \ref{PropMM} and renormalization of the action function of 1D gravity, one has:
\begin{align*}
&\int_{\mathbb{H}_N}dM \exp{\left(\frac{1}{g_s}tr\left(-\frac{1}{2}M^2+\sum_{n=0}^{\infty}t_n\frac{M^{n+1}}{(n+1)!}\right)\right)}\\
=&\int_{\mathbb{R}^{N}}\prod_{1\leqslant i<j\leqslant N}(\lambda_i-\lambda_{j})^2
\exp{\left(\frac{1}{g_s}\sum_{i=1}^{N}
\left(-\frac{1}{2}\lambda_i^2+\sum_{n=0}^{\infty}t_n\frac{\lambda_i^{n+1}}{(n+1)!}\right)\right)}
\prod_{i=1}^{N}d\lambda_i\\
=&\int_{\mathbb{R}^{N}}\prod_{1\leqslant i<j\leqslant N}(\lambda_i-I_0-(\lambda_{j}-I_0))^2
\exp\left(\frac{1}{g_s}\sum_{i=1}^{N}
\sum_{k=0}^{\infty}\frac{(-1)^k}{(k+1)!}(I_k+\delta_{k,1})I_{0}^{k+1}\right)\cdot\\
&\exp\left(\frac{1}{g_s}\sum_{i=1}^{N}
\left(-\frac{1}{2}(\lambda_i-I_0)^2+\sum_{n=1}^{\infty}I_n\frac{(\lambda_i-I_0)^{n+1}}{(n+1)!}\right)\right)
\prod_{i=1}^{N}d\lambda_i\\
=&\exp\left(\frac{N}{g_s}\sum_{k=0}^{\infty}\frac{(-1)^k}{(k+1)!}(I_k+\delta_{k,1})I_{0}^{k+1}\right)\cdot\\
&\int_{\mathbb{R}^{N}}\prod_{1\leqslant i<j\leqslant N}(\lambda_i-\lambda_{j})^2
\exp\left(\frac{1}{g_s}\sum_{i=1}^{N}
\left(-\frac{1}{2}\lambda_i^2+\sum_{n=1}^{\infty}I_n\frac{\lambda_i^{n+1}}{(n+1)!}\right)\right)
\prod_{i=1}^{N}d\lambda_i\\
=&\exp\left(\frac{N}{g_s}\sum_{k=0}^{\infty}\frac{(-1)^k}{(k+1)!}(I_k+\delta_{k,1})I_{0}^{k+1}\right)\cdot
\int_{\mathbb{H}_N}dM
\exp\left(\frac{1}{g_s}tr
\left(-\frac{1}{2}M^2+\sum_{n=1}^{\infty}I_n\frac{M^{n+1}}{(n+1)!}\right)\right).
\end{align*}
\end{proof}
The free energy $F^{N}$ of Hermitian one-matrix models is defined by:
\be
F^N:=\log{Z^{N}}.
\ee
There is a genus expansions for $F^N$:
\be\label{Ngenusexpansion}
F^{N}=\sum_{g=0}^{\infty}g_s^{g-1}F^{N}_g.
\ee
where $\{F_g^{N}\}_{g\geqslant0}$ are formal power series of $t_0,t_1, \cdots$.
This is called  the thin genus expansion in \cite{Zhou3}.
By a result in \cite{Zhou3}, if we define
\be
\deg{t_n}=n-1, \ \ \ n=0,1,2,\cdots
\ee
then $F^{N}_{g}$ is weighted homogeneous in $t_0,t_1, \cdots$ with
\be
\deg{F_{g}^{1D}}=2g-2,\ \ \ g=0,1,2,\cdots
\ee
The following analogue of the Itzykson-Zuber Ansatz and Theorem \ref{thm:F1D}
is proved in \cite{Zhou3}:

\begin{Theorem}(\cite[Theorem 5.1 and 5.2]{Zhou3})
\begin{align}
F^{N}_0=&N\sum_{k=0}^{\infty}\frac{(-1)^k}{(k+1)!}(I_k+\delta_{k,1})I_{0}^{k+1},\label{eqn:F0N} \\
F^{N}_1=&\frac{N^2}{2}\log{\frac{1}{1-I_1}},\label{eqn:F1N}\\
F^{N}_g=&\sum_{\sum_{2\leqslant k\leqslant 2g-1}m_k(k-1)=2g-2} \left<\prod_{j=2}^{2g-1}\tau_{j}^{m_j}\right>^{N}_g\cdot \prod_{j=2}^{2g-1}\frac{1}{m_j!}\left(\frac{I_{j}}{(1-I_1)^{\frac{j+1}{2}}}\right)^{m_j} , \ g\geqslant2 \label{eqn:FgN}
\end{align}
where the correlators are defined by:
\be
\left<\tau_{a_1}\cdots\tau_{a_n}\right>^{N}_g=\frac{\partial^n F^{N}_g}{\partial t_{a_1}\cdots \partial t_{a_n}}\bigg|_{t=0}.
\ee
\end{Theorem}

In the literature another type of genus expansion is used.
By introducing the 't Hooft coupling constant
\be
t = N g_s,
\ee
$F^N$ can be rewritten as:
\be\label{tgenusexpansion}
F^{N}=\sum_{g=0}^{\infty}g_s^{2g-2}F^{t}_g.
\ee
where $\{F_g^{t}\}_{g\geqslant0}$ are formal power series of $t_0,t_1, \cdots$ and $t$.
This is called  the fat genus expansion in \cite{Zhou3}.
We will study both thin and fat genus expansions for the free energy function of Hermitian one-matrix
model by the correspondence Virasoro constraints.

\subsection{Virasoro constraints of Hermitian one-matrix model}
\begin{Theorem}(Virasoro constraints for thin genus expansion\cite{Zhou3})
The partition function $Z^{N}$ of Hermitian one-matrix model satisfies the following equations for $m\geqslant-1$:
\be
L^{N}_mZ^{N}=0,
\ee
where
\begin{align}
L^{N}_{-1}&=\frac{Nt_0}{g_s}+\sum_{n\geqslant 1}(t_n-\delta_{n,1})\frac{\partial}{\partial t_{n-1}},\\
L^{N}_{0}&=N^2+\sum_{n\geqslant 0}(n+1)(t_n-\delta_{n,1})\frac{\partial}{\partial t_{n}},\\
L^{N}_{m}&=2Ng_sm!\frac{\partial}{\partial t_{m-1}}+\sum_{n\geqslant 0}\frac{(m+n+1)!}{n!}(t_n-\delta_{n,1})\frac{\partial}{\partial t_{m+n}}
+g_s^2\sum_{k=1}^{m-1}k!(m-k)!\frac{\partial}{\partial t_{k-1}}\frac{\partial}{\partial t_{m-k-1}}.
\end{align}
for $m\geqslant1$. Furthermore, $\{L^{N}_m\}_{m\geqslant-1}$ satisfies the following commutation relations:
\be
\left[L^{N}_m,L^{N}_n\right]=(m-n)L^{N}_{m+n}.
\ee
for $m,n\geqslant-1$.
\end{Theorem}
\begin{Theorem}(Virasoro constraints for fat genus expansion\cite{Zhou3})
The partition function $Z^{N}$ of Hermitian one-matrix model with $t=Ng_s$ satisfies the following equations for $m\geqslant-1$:
\be
L^{t}_mZ^{N}=0,
\ee
where
\begin{align}
L^{t}_{-1}&=\frac{tt_0}{g_s^2}+\sum_{n\geqslant 1}(t_n-\delta_{n,1})\frac{\partial}{\partial t_{n-1}},\\
L^{t}_{0}&=\frac{t^2}{g_s^2}+\sum_{n\geqslant 0}(n+1)(t_n-\delta_{n,1})\frac{\partial}{\partial t_{n}},\\
L^{t}_{m}&=2tm!\frac{\partial}{\partial t_{m-1}}+\sum_{n\geqslant 0}\frac{(m+n+1)!}{n!}(t_n-\delta_{n,1})\frac{\partial}{\partial t_{m+n}}
+g_s^2\sum_{k=1}^{m-1}k!(m-k)!\frac{\partial}{\partial t_{k-1}}\frac{\partial}{\partial t_{m-k-1}}.
\end{align}
for $m\geqslant1$. Furthermore, $\{L^{t}_m\}_{m\geqslant-1}$ satisfies the following commutation relations:
\be
\left[L^{t}_m,L^{t}_n\right]=(m-n)L^{t}_{m+n}.
\ee
for $m,n\geqslant-1$.
\end{Theorem}
Similarly as in the topological 1D gravity theory, one can rewrite the Virasoro operators in I-coordinates:
\begin{Theorem} The Virasoro operators for thin genus expansion can be written
in $I$-coordinates as follows:
\begin{align}
L^{N}_{-1}=&-\frac{\partial}{\partial I_0} +\frac{N}{g_s}\sum_{n=0}^{\infty}\frac{(-1)^nI_0^n}{n!}I_n,\label{NIstring}\\
L^{N}_0=&-I_0\frac{\partial}{\partial I_0}-2\frac{\partial}{\partial I_1}+\sum_{l\geqslant 1}(l+1)I_l\frac{\partial}{\partial I_l}+N^2, \label{NIdilaton}\\
L_{m}^N=&2Ng_sm!\left(\frac{1}{1-I_1}\frac{I_0^{m-1}}{(m-1)!}d_{X} +\sum_{1\leqslant l \leqslant m-1}\frac{I_0^{m-1-l}}{(m-1-l)!}\frac{\partial}{\partial I_l}\right)\\
&-I_{0}^{m+1}\frac{\partial}{\partial I_0} +\sum_{i=1}^{m+1}\binom{m+1}{i}I_0^{m+1-i} \sum_{p=1}^{\infty}\frac{(p+i)!}{p!}(I_{p}-\delta_{p,1})\frac{\partial}{\partial I_{p+i-1}}\nonumber\\
&+g_s^2\binom{m+1}{3}\left(\frac{(m-2)I_0^{m-3}}{1-I_1}d_{X}\frac{\partial}{\partial I_{1}} +\frac{I_0^{m-2}}{(1-I_1)^2}d_{X}^2 +\left(\frac{(m-2)I_0^{m-3}}{(1-I_1)^2}+\frac{I_0^{m-2}I_2}{(1-I_1)^3}\right) d_{X}\right) \nonumber\\
&+2g_s^2\sum_{j=2}^{m-2}\binom{m+1}{j+3}(j+1)!\frac{I_{0}^{m-2-j}}{1-I_1} \left(\frac{\partial}{\partial I_{j-1}}+d_{X}\frac{\partial}{\partial I_j}\right)\nonumber\\
&+g_s^2\sum_{i=1}^{m-2}\sum_{j=1}^{m-2-i}\binom{m+1}{i+j+3}(i+1)!(j+1)!I_0^{m-2-i-j}\frac{\partial}{\partial I_i}\frac{\partial}{\partial I_j},\nonumber
\end{align}
where
\be
d_{X}=\frac{\partial}{\partial I_0}+\sum_{l\geqslant1}I_{l+1}\frac{\partial}{\partial I_l}.
\ee
Moreover, by taking $N=\frac{t}{g_s}$, one get the Virasoro operators for fat genus expansion
in $I$-coordinates:
\begin{align}
L^{t}_{-1}=&-\frac{\partial}{\partial I_0} +\frac{t}{g_s^2}\sum_{n=0}^{\infty}\frac{(-1)^nI_0^n}{n!}I_n,\label{tIstring}\\
L^{t}_0=&-I_0\frac{\partial}{\partial I_0}-2\frac{\partial}{\partial I_1}+\sum_{l\geqslant 1}(l+1)I_l\frac{\partial}{\partial I_l}+\frac{t^2}{g_s^2}, \label{tIdilaton}\\
L_{m}^t=&2tm!\left(\frac{1}{1-I_1}\frac{I_0^{m-1}}{(m-1)!}d_{X} +\sum_{1\leqslant l \leqslant m-1}\frac{I_0^{m-1-l}}{(m-1-l)!}\frac{\partial}{\partial I_l}\right)\\
&-I_{0}^{m+1}\frac{\partial}{\partial I_0} +\sum_{i=1}^{m+1}\binom{m+1}{i}I_0^{m+1-i} \sum_{p=1}^{\infty}\frac{(p+i)!}{p!}(I_{p}-\delta_{p,1})\frac{\partial}{\partial I_{p+i-1}}\nonumber\\
&+g_s^2\binom{m+1}{3}\left(\frac{(m-2)I_0^{m-3}}{1-I_1}d_{X}\frac{\partial}{\partial I_{1}} +\frac{I_0^{m-2}}{(1-I_1)^2}d_{X}^2 +\left(\frac{(m-2)I_0^{m-3}}{(1-I_1)^2}+\frac{I_0^{m-2}I_2}{(1-I_1)^3}\right) d_{X}\right) \nonumber\\
&+2g_s^2\sum_{j=2}^{m-2}\binom{m+1}{j+3}(j+1)!\frac{I_{0}^{m-2-j}}{1-I_1} \left(\frac{\partial}{\partial I_{j-1}}+d_{X}\frac{\partial}{\partial I_j}\right)\nonumber\\
&+g_s^2\sum_{i=1}^{m-2}\sum_{j=1}^{m-2-i}\binom{m+1}{i+j+3}(i+1)!(j+1)!I_0^{m-2-i-j}\frac{\partial}{\partial I_i}\frac{\partial}{\partial I_j},\nonumber
\end{align}
\end{Theorem}
\begin{proof}
We have proved the following identity in the proof of theorem \ref{1Dcase}:
\begin{align}
\sum_{n\geqslant 0}\frac{(m+n+1)!}{n!}(t_n-\delta_{n,1})\frac{\partial}{\partial t_{m+n}}
=&-I_{0}^{m+1}\frac{\partial}{\partial I_0} +\sum_{i=1}^{m+1}\binom{m+1}{i}I_0^{m+1-i} \sum_{p=1}^{\infty}\frac{(p+i)!}{p!}(I_{p}-\delta_{p,1})\frac{\partial}{\partial I_{p+i-1}}.
\end{align}
Now we rewrite $\sum_{k=1}^{m-1}k!(m-k)!\frac{\partial}{\partial t_{k-1}}\frac{\partial}{\partial t_{m-k-1}}$ in I-coordinates:

\begin{align*}
&\sum_{k=1}^{m-1}k!(m-k)!\frac{\partial}{\partial t_{k-1}}\frac{\partial}{\partial t_{m-k-1}}\\
=&\sum_{k=1}^{m-1}k!(m-k)! \left(\frac{1}{1-I_1}\frac{I_0^{k-1}}{(k-1)!}
\left(\frac{\partial}{\partial I_0}+\sum_{l\geqslant1}I_{l+1}\frac{\partial}{\partial I_l}\right) +\sum_{1\leqslant i \leqslant k-1}\frac{I_0^{k-1-i}}{(k-1-i)!}\frac{\partial}{\partial I_i}\right)
\\
&\cdot \left(\frac{1}{1-I_1}\frac{I_0^{m-k-1}}{(m-k-1)!}
\left(\frac{\partial}{\partial I_0}+\sum_{l\geqslant1}I_{l+1}\frac{\partial}{\partial I_l}\right)
+\sum_{1\leqslant j \leqslant m-k-1}\frac{I_0^{m-k-1-j}}{(m-k-1-j)!}
\frac{\partial}{\partial I_j}\right)\\
=&\sum_{k=1}^{m-1} \left(\frac{kI_0^{k-1}}{1-I_1}
\left(\frac{\partial}{\partial I_0}+\sum_{l\geqslant1}I_{l+1}\frac{\partial}{\partial I_l}\right)
 +\sum_{1\leqslant i \leqslant k-1}\frac{k!I_0^{k-1-i}}{(k-1-i)!}
 \frac{\partial}{\partial I_i}\right)\\
&\cdot \left(\frac{(m-k)I_0^{m-k-1}}{1-I_1}
\left(\frac{\partial}{\partial I_0}+\sum_{l\geqslant1}I_{l+1}\frac{\partial}{\partial I_l}\right)
+\sum_{1\leqslant j \leqslant m-k-1}\frac{(m-k)!I_0^{m-k-1-j}}{(m-k-1-j)!}
\frac{\partial}{\partial I_j}\right)\\
=&\sum_{k=1}^{m-1}k(m-k)\frac{I_0^{m-2}}{(1-I_1)^2}
\left(\frac{\partial}{\partial I_0}+\sum_{l\geqslant1}I_{l+1}\frac{\partial}{\partial I_l}\right)^2
\\
&+\sum_{k=1}^{m-1}\sum_{1\leqslant j \leqslant m-k-1}\frac{k(m-k)!}{(m-k-1-j)!}
\frac{I_0^{m-2-j}}{1-I_1}
\left(\frac{\partial}{\partial I_0}+\sum_{l\geqslant1}I_{l+1}\frac{\partial}{\partial I_l}\right)
\frac{\partial}{\partial I_j}\\
&+\sum_{k=1}^{m-1}\sum_{1\leqslant i \leqslant k-1}
\frac{(m-k)k!}{(k-1-i)!}\frac{I_0^{m-2-i}}{1-I_1} \frac{\partial}{\partial I_i}
\left(\frac{\partial}{\partial I_0}+\sum_{l\geqslant1}I_{l+1}\frac{\partial}{\partial I_l}\right)
\\
&+\sum_{k=1}^{m-1}\sum_{i=1}^{k-1}\sum_{j=1}^{m-k-1}
\frac{(m-k)!k!}{(k-1-i)!(m-k-1-j)!}I_0^{m-2-i-j} \frac{\partial}{\partial I_i}
\frac{\partial}{\partial I_j} \\
&+\sum_{k=1}^{m-1}k(m-k)\frac{I_0^{m-3}}{(1-I_1)^3}\left((m-k-1)(1-I_1)+I_0I_2\right)
\left(\frac{\partial}{\partial I_0}+\sum_{l\geqslant1}I_{l+1}\frac{\partial}{\partial I_l}\right)
 \\
&+\sum_{k=1}^{m-1}\sum_{1\leqslant j \leqslant m-k-1}\frac{k(m-k)!}{(m-k-2-j)!}
\frac{I_0^{m-3-j}}{1-I_1} \frac{\partial}{\partial I_j}\\
&+\sum_{k=2}^{m-1}\frac{(m-k)k!}{(k-2)!}\frac{I_0^{m-3}}{(1-I_1)^2}
\left(\frac{\partial}{\partial I_0}+\sum_{l\geqslant1}I_{l+1}\frac{\partial}{\partial I_l}\right).
\end{align*}
Denote these summations by (a)-(g) respectively, then using the following identity:
\begin{equation}
\sum_{k=a}^{m-b}\frac{k!(m-k)!}{(k-a)!(m-k-b)!}=\binom{m+1}{a+b+1}a!b!,
\end{equation}
we have:
\begin{align}
(a)=&\binom{m+1}{3}\frac{I_0^{m-2}}{(1-I_1)^2}
\left(\frac{\partial}{\partial I_0}+\sum_{l\geqslant1}I_{l+1}\frac{\partial}{\partial I_l}\right)^2,
\\
(b)=&\sum_{j=1}^{m-2}\binom{m+1}{j+3}I_{0}^{m-2-j}\frac{1}{1-I_1}
\left(\frac{\partial}{\partial I_0}+\sum_{l\geqslant1}I_{l+1}\frac{\partial}{\partial I_l}\right)
(j+1)!\frac{\partial}{\partial I_j},\\
(c)=&\sum_{i=1}^{m-2}\binom{m+1}{i+3}I_{0}^{m-2-i}\frac{1}{1-I_1}(i+1)!
\frac{\partial}{\partial I_i}
\left(\frac{\partial}{\partial I_0}+\sum_{l\geqslant1}I_{l+1}
\frac{\partial}{\partial I_l}\right)\\
=&\sum_{i=1}^{m-2}\binom{m+1}{i+3}I_{0}^{m-2-i}\frac{1}{1-I_1}
\left(\frac{\partial}{\partial I_0}+\sum_{l\geqslant1}I_{l+1}
\frac{\partial}{\partial I_l}\right)(i+1)!\frac{\partial}{\partial I_i}\nonumber\\
&+\sum_{i=2}^{m-2}\binom{m+1}{i+3}I_{0}^{m-2-i}\frac{1}{1-I_1}(i+1)!
\frac{\partial}{\partial I_{i-1}},\nonumber\\
(d)=&\sum_{i=1}^{m-2}\sum_{j=1}^{m-2-i}\sum_{k=i+1}^{m-j-1}
\frac{(m-k)!k!}{(k-1-i)!(m-k-1-j)!}I_0^{m-2-i-j} \frac{\partial}{\partial I_i}
\frac{\partial}{\partial I_j}\\
=&\sum_{i=1}^{m-2}\sum_{j=1}^{m-2-i}\binom{m+1}{i+j+3}I_0^{m-2-i-j}(i+1)!(j+1)!
\frac{\partial}{\partial I_i}\frac{\partial}{\partial I_j},\nonumber\\
(e)=&\left(2\binom{m+1}{4}\frac{I_0^{m-3}}{(1-I_1)^2}+\binom{m+1}{3}
\frac{I_0^{m-2}I_2}{(1-I_1)^3}\right) \left(\frac{\partial}{\partial I_0}+\sum_{l\geqslant1}I_{l+1}\frac{\partial}{\partial I_l}\right),\\
(f)=&\sum_{j=1}^{m-3}\binom{m+1}{j+4}\frac{I_0^{m-3-j}}{1-I_1} (j+2)!
\frac{\partial}{\partial I_j},\\
(g)=&2\binom{m+1}{4}\frac{I_0^{m-3}}{(1-I_1)^2}
\left(\frac{\partial}{\partial I_0}+\sum_{l\geqslant1}I_{l+1}\frac{\partial}{\partial I_l}\right).
\end{align}
Plus all these equations together, we complete the proof.
\end{proof}

\subsection{Computations of $F^{N}_g$ by Virasoro constraints for thin genus expansion in $I$-coordinates}
As applications of the expressions of the Virasoro operators expressed in $I$-coordinates
derived in last Subsection,
we use them in this Subsection to compute $F^{N}_g$.
\begin{Theorem}\label{HMM}
For the free energies $\{F_{g}^N\}_{g\geqslant2}$ of the Hermitian one-matrix model, the following equations hold:

\begin{align}
\frac{\partial F^{N}_{g}}{\partial I_1}=&\frac{1}{2(1-I_1)}\sum_{n=2}^{2g-1}(n+1)I_n\frac{\partial F^{N}_{g}}{\partial I_n},\label{d}\\
\frac{\partial F^N_{g}}{\partial I_{2}}=& \frac{N}{3(1-I_1)}\frac{1}{1-I_1} d_{X}\left(F^N_{g-1}\right) +\sum_{n=2}^{2g-2} \frac{(n+2)!}{3!n!}\frac{I_n}{1-I_1}\frac{\partial F^N_{g}}{\partial I_{n+1}},\\
\frac{\partial F^N_{g}}{\partial I_{3}}=& \frac{N}{6(1-I_1)}\frac{\partial F^N_{g-1}}{\partial I_1} +\sum_{n=2}^{2g-3}\frac{(n+3)!}{4!n!}\frac{I_n}{1-I_1}\frac{\partial F^N_{g}}{\partial I_{2+n}}+\frac{1}{4!(1-I_1)^3}\sum_{g_1=1}^{g-2}d_{X}\left(F^N_{g_1}\right) d_{X}\left(F^N_{g-1-g_1}\right)  \\
&+\frac{1}{4!(1-I_1)}\left(\frac{1}{1-I_1}d_{X}\right)^2\left(F_{g-2}^N\right) +\delta_{g,2}\frac{N}{4!(1-I_1)^2},\nonumber\\
\frac{\partial F^N_{g}}{\partial I_{4}} =&\frac{N}{10(1-I_1)}\frac{\partial F^N_{g-1}}{\partial I_{2}}+\sum_{n=2}^{2g-4}\frac{(n+4)!}{5!n!}\frac{I_n}{1-I_1}\frac{\partial F^N_{g}}{\partial I_{n+3}}+\frac{1}{30(1-I_1)^2}\sum_{g_1=1}^{g-2}\frac{\partial F^N_{g_1}}{\partial I_{1}} d_{X}\left(F^N_{g-1-g_1}\right)\label{e}\\
&  +\frac{1}{60(1-I_1)^3} \left(d_{X}\left(F^N_{g-2}\right)+2(1-I_1)d_{X}\left(\frac{\partial F^N_{g-2}}{\partial I_{1}}\right) \right),\nonumber\\
\frac{\partial F^N_{g}}{\partial I_{p+4}}=&\frac{2N(p+3)!}{(p+5)!(1-I_1)}\frac{\partial F^N_{g-1}}{\partial I_{p+2}}+\sum_{n=2}^{2g-4-p}\frac{(p+n+4)!}{n!(p+5)!}\frac{I_n}{1-I_1}\frac{\partial F^N_{g}}{\partial I_{p+n+3}}\label{f}\\
&+\sum_{g_1=1}^{g-2}\sum_{k=2}^{p+1}\frac{k!(p+3-k)!}{(p+5)!(1-I_1)}\frac{\partial F^N_{g_1}}{\partial I_{k-1}}\frac{\partial F^N_{g-1-g_1}}{\partial I_{p+2-k}} +\frac{2}{(1-I_1)^2}\frac{(p+2)!}{(p+5)!}\sum_{g_1=1}^{g-2}\frac{\partial F^N_{g_1}}{\partial I_{p+1}} d_{X}\left(F^N_{g-1-g_1}\right)\nonumber\\
& +\sum_{k=2}^{p+1}\frac{k!(p+3-k)!}{(p+5)!(1-I_1)}\frac{\partial^2 F^N_{g-2}}{\partial I_{k-1}\partial I_{p+2-k}} +\frac{1}{(1-I_1)^2}\frac{(p+2)!}{(p+5)!} \left(\frac{\partial F^N_{g-2}}{\partial I_{p}}+2d_{X}\left(\frac{\partial F^N_{g-2}}{\partial I_{p+1}}\right) \right),\nonumber
\end{align}
where $p=1,2,\cdots,2g-5$.
\end{Theorem}

\begin{proof}
The first equation has been proved in \cite{Zhou3}. For the rest equations,
we first rewrite Virasoro constraints for partition function into Virasoro constraints for free energy:
\begin{align}
0=(Z^{N})^{-1}L^{N}_{m}Z^{N}=&2Ng_sm!\frac{\partial F^{N}}{\partial t_{m-1}}+\sum_{n\geqslant 0}\frac{(m+n+1)!}{n!}(t_n-\delta_{n,1})\frac{\partial F^{N}}{\partial t_{m+n}}  \\
& +g_s^2\sum_{k=1}^{m-1}k!(m-k)!\frac{\partial F^{N}}{\partial t_{k-1}}\frac{\partial F^{N}}{\partial t_{m-k-1}}+g_s^2\sum_{k=1}^{m-1}k!(m-k)!\frac{\partial^2 F^{N}}{\partial t_{k-1}\partial t_{m-k-1}},
\nonumber
\end{align}
Let
\be
\tilde{F}^{N}=\sum_{g\geqslant1}g_{s}^{g-1}F^{N}_g,
\ee
then $\tilde{F}^{N}$ does not depend on $I_0$, and
\begin{align}
0=&2Nm!\frac{\partial F_0^{N}}{\partial t_{m-1}}+
\frac{1}{g_s}\sum_{n\geqslant 0}\frac{(m+n+1)!}{n!}(t_n-\delta_{n,1})
\frac{\partial F_0^{N}}{\partial t_{m+n}}
 +\sum_{k=1}^{m-1}k!(m-k)!\frac{\partial F_0^{N}}{\partial t_{k-1}}
\frac{\partial F_0^{N}}{\partial t_{m-k-1}}\label{eqn:VirasoroF}\\
&+2g_s\sum_{k=1}^{m-1}k!(m-k)!\frac{\partial F_0^{N}}{\partial t_{k-1}}
\frac{\partial \tilde{F}^{N}}{\partial t_{m-k-1}}
+g_s\sum_{k=1}^{m-1}k!(m-k)!
\frac{\partial^2 F_0^{N}}{\partial t_{k-1}\partial t_{m-k-1}},\nonumber\\
&+2Ng_sm!\frac{\partial \tilde{F}^{N}}{\partial t_{m-1}}+
\sum_{n\geqslant 0}\frac{(m+n+1)!}{n!}(t_n-\delta_{n,1})
\frac{\partial \tilde{F}^{N}}{\partial t_{m+n}}  \nonumber\\
& +g_s^2\sum_{k=1}^{m-1}k!(m-k)!\frac{\partial \tilde{F}^{N}}{\partial t_{k-1}}
\frac{\partial \tilde{F}^{N}}{\partial t_{m-k-1}}
+g_s^2\sum_{k=1}^{m-1}k!(m-k)!\frac{\partial^2 \tilde{F}^{N}}{\partial t_{k-1}\partial t_{m-k-1}},
\nonumber
\end{align}
Now we let $I_0=0$, under this condition, one has:
\begin{align*}
t_0=&0, \\
 t_k=&I_k, \ \ \ k>0
\end{align*}
and
\begin{align}
\frac{\partial}{\partial t_0}&=\frac{1}{1-I_1}\frac{\partial}{\partial I_0}+\sum_{l\geqslant1}\frac{I_{l+1}}{1-I_1}\frac{\partial}{\partial I_l},\\
\frac{\partial}{\partial t_k}&=\frac{\partial}{\partial I_k}, \ \ \ k>0
\end{align}
By \eqref{eqn:F0N}
\begin{align}
\frac{\partial F_0^N}{\partial I_0}&=Nt_0,\\
\frac{\partial F_0^N}{\partial I_k}&=N\frac{(-1)^kI_{0}^{k+1}}{(k+1)!}, \ \ k\geqslant1.
\end{align}
So one has:
\be
\frac{\partial F_0^N}{\partial I_0}\bigg|_{I_0=0}=\frac{\partial F_0^N}{\partial I_k}\bigg|_{I_0=0}=0,
\ee
and therefore,
\be
\frac{\partial F_0^N}{\partial t_0}\bigg|_{I_0=0}=\frac{\partial F_0^N}{\partial t_k}\bigg|_{I_0=0}=0.
\ee
Moreover,
\begin{align}
\frac{\partial^2 F_0^N}{\partial I_0^2}=&N(1-t_1),\\
\frac{\partial^2 F_0^N}{\partial I_0\partial I_k}=&N\frac{(-1)^kI_{0}^{k}}{k!}, \ \ k\geqslant1\\
\frac{\partial^2 F_0^N}{\partial I_k\partial I_l}=&0, \ \ k,l\geqslant1.
\end{align}
When restrict $I_0=0$, these equations give
\begin{align}
\frac{\partial^2 F_0^N}{\partial I_0^2}\bigg|_{I_0=0}=&N(1-I_1),\\
\frac{\partial^2 F_0^N}{\partial I_0\partial I_k}\bigg|_{I_0=0}=&
\frac{\partial^2 F_0^N}{\partial I_k\partial I_l}\bigg|_{I_0=0}=0, \ \ k,l\geqslant1.
\end{align}
Therefore
\begin{align}
\frac{\partial^2 F_0^N}{\partial t_0^2}\bigg|_{I_0=0}=&\left(\frac{1}{1-I_1}\frac{\partial}{\partial I_0}+\sum_{l\geqslant1}\frac{I_{l+1}}{1-I_1}\frac{\partial}{\partial I_l}\right)^2(F_0^N)\bigg|_{I_0=0}\\
=&\frac{1}{(1-I_1)^2}\left(\frac{\partial}{\partial I_0}+\sum_{l\geqslant1}I_{l+1}\frac{\partial}{\partial I_l}\right)^2(F_0^N)\bigg|_{I_0=0} +\frac{I_2}{(1-I_1)^3}\left(\frac{\partial}{\partial I_0}+\sum_{l\geqslant1}I_{l+1}\frac{\partial}{\partial I_l}\right)(F_0^N)\bigg|_{I_0=0}\nonumber\\
=&\frac{1}{(1-I_1)^2}\frac{\partial^2 F_0^N}{\partial I_0^2}\bigg|_{I_0=0}\nonumber\\
=&\frac{N}{1-I_1},\nonumber\\
\frac{\partial^2 F_0^N}{\partial t_0\partial t_k}\bigg|_{I_0=0}=&\left(\frac{1}{1-I_1}\frac{\partial}{\partial I_0}+\sum_{l\geqslant1}\frac{I_{l+1}}{1-I_1}\frac{\partial}{\partial I_l}\right)\left(\frac{\partial F_0^N}{\partial I_k}\right)\bigg|_{I_0=0}\\
=&\left(\frac{1}{1-I_1}\frac{\partial}{\partial I_0}+\sum_{l\geqslant1}\frac{I_{l+1}}{1-I_1}\frac{\partial}{\partial I_l}\right)\left(N\frac{(-1)^kI_{0}^{k+1}}{(k+1)!}\right)\bigg|_{I_0=0}\nonumber\\
=&0 \ \ k\geqslant1,\nonumber\\
\frac{\partial^2 F_0^N}{\partial t_k\partial t_l}\bigg|_{I_0=0}=&\frac{\partial^2 F_0^N}{\partial I_k\partial I_l}\bigg|_{I_0=0}=0, \ \ k,l\geqslant1.
\end{align}
Hence \eqref{eqn:VirasoroF} becomes
\begin{align}
0=&\delta_{m,2}\frac{Ng_s}{1-I_1}+2Ng_sm!\frac{\partial \tilde{F}^{N}}{\partial t_{m-1}}+
\sum_{n\geqslant 1}\frac{(m+n+1)!}{n!}(t_n-\delta_{n,1})
\frac{\partial \tilde{F}^{N}}{\partial t_{m+n}}  \\
& +g_s^2\sum_{k=1}^{m-1}k!(m-k)!\frac{\partial \tilde{F}^{N}}{\partial t_{k-1}}
\frac{\partial \tilde{F}^{N}}{\partial t_{m-k-1}}
+g_s^2\sum_{k=1}^{m-1}k!(m-k)!\frac{\partial^2 \tilde{F}^{N}}{\partial t_{k-1}\partial t_{m-k-1}},\nonumber
\end{align}
or
\begin{align}
0=&\delta_{m,2}\delta_{g,1}\frac{N}{1-I_1}+2Nm!\frac{\partial F_g^{N}}{\partial t_{m-1}}+
\sum_{n\geqslant 0}\frac{(m+n+1)!}{n!}(t_n-\delta_{n,1})
\frac{\partial F_{g+1}^{N}}{\partial t_{m+n}}\\
& +\sum_{g_1+g_2=g}\sum_{k=1}^{m-1}k!(m-k)!\frac{\partial F_{g_1}^{N}}{\partial t_{k-1}}
\frac{\partial F_{g_2}^{N}}{\partial t_{m-k-1}}
+\sum_{k=1}^{m-1}k!(m-k)!\frac{\partial^2 F_{g-1}^{N}}{\partial t_{k-1}\partial t_{m-k-1}},  \nonumber
\end{align}
with
\begin{align}
\frac{\partial}{\partial t_0}&=\frac{1}{1-I_1}d_{X},\\
\frac{\partial}{\partial t_k}&=\frac{\partial}{\partial I_k}, \ \ \ k>0
\end{align}
for $m=1$:
\be
\begin{split}
0=&2N\sum_{l\geqslant1}\frac{I_{l+1}}{1-I_1}\frac{\partial F^N_{g}}{\partial I_l} +\sum_{n\geqslant 1}\frac{(n+2)!}{n!}(I_n-\delta_{n,1})\frac{\partial F^N_{g+1}}{\partial I_{n+1}},
\end{split}
\ee
for $m=2$:
\be
\begin{split}
0=&4N\frac{\partial F^N_{g}}{\partial I_1} +\sum_{n\geqslant 1}\frac{(n+3)!}{n!}(I_n-\delta_{n,1})\frac{\partial F^N_{g+1}}{\partial I_{2+n}}+\sum_{g_1=1}^{g-1}\sum_{l\geqslant1}\frac{I_{l+1}}{1-I_1}\frac{\partial F^N_{g_1}}{\partial I_l} \sum_{k\geqslant1}\frac{I_{k+1}}{1-I_1}\frac{\partial F^N_{g-g_1}}{\partial I_k} \\
&+\left(\sum_{l\geqslant1}\frac{I_{l+1}}{1-I_1}\frac{\partial}{\partial I_l}\right)^2\left(F_{g-1}^N\right)+\delta_{g,1}\frac{N}{1-I_1},
\end{split}
\ee
for $m=3$:
\be
\begin{split}
0=&12N\frac{\partial F^N_{g}}{\partial I_{2}}+\sum_{n\geqslant 0}\frac{(n+4)!}{n!}(I_n-\delta_{n,1})\frac{\partial F^N_{g+1}}{\partial I_{3+n}} +4\sum_{g_1=1}^{g-1}\frac{\partial F^N_{g_1}}{\partial I_{1}}\sum_{l\geqslant1}\frac{I_{l+1}}{1-I_1}\frac{\partial F^N_{g-g_1}}{\partial I_l} \\
&+4\sum_{l\geqslant1}\frac{I_{l+1}}{1-I_1}\frac{\partial^2 F^N_{g-1}}{\partial I_l\partial I_{1}} +2\sum_{l\geqslant1}\frac{I_{l+1}}{(1-I_1)^2}\frac{\partial F^N_{g-1}}{\partial I_l},
\end{split}
\ee
for $m\geqslant4$:
\be
\begin{split}
0=&2Nm!\frac{\partial F^N_{g}}{\partial I_{m-1}}+\sum_{n\geqslant 0}\frac{(m+n+1)!}{n!}(I_n-\delta_{n,1})\frac{\partial F^N_{g+1}}{\partial I_{m+n}}  +\sum_{g_1=1}^{g-1}\sum_{k=2}^{m-2}k!\frac{\partial F^N_{g_1}}{\partial I_{k-1}}(m-k)!\frac{\partial F^N_{g-g_1}}{\partial I_{m-k-1}} \\
&+2\sum_{g_1=1}^{g-1}(m-1)!\frac{\partial F^N_{g_1}}{\partial I_{m-2}}\sum_{l\geqslant1}\frac{I_{l+1}}{1-I_1}\frac{\partial F^N_{g-g_1}}{\partial I_l} +\sum_{k=2}^{m-2}k!(m-k)!\frac{\partial^2 F^N_{g-1}}{\partial I_{k-1}\partial I_{m-k-1}} \\
&+2(m-1)!\sum_{l\geqslant1}\frac{I_{l+1}}{1-I_1}\frac{\partial^2 F^N_{g-1}}{\partial I_l\partial I_{m-2}} +(m-1)!\frac{1}{1-I_1}\frac{\partial F^N_{g-1}}{\partial I_{m-3}}.
\end{split}
\ee
The proof is completed.
\end{proof}

Now we explain how to use Theorem \ref{HMM} to calculate free energies of higher genus.
By \eqref{eqn:FgN}, $F^{N}_g$ depends only on $I_1, \dots, I_{2g-1}$ for $g\geqslant 2$.
By \eqref{f}, we have:

\begin{align}
\frac{\partial F^N_{g}}{\partial I_{2g-1}}=&\frac{2N(2g-2)!}{(2g)!(1-I_1)}\frac{\partial F^N_{g-1}}{\partial I_{2g-3}}+\frac{1}{(1-I_1)^2}\frac{(2g-3)!}{(2g)!} \frac{\partial F^N_{g-2}}{\partial I_{2g-5}} +\sum_{g_1=1}^{g-2}\sum_{k=2}^{2g-4}\frac{k!(2g-2-k)!}{(2g)!(1-I_1)}\frac{\partial F^N_{g_1}}{\partial I_{k-1}}\frac{\partial F^N_{g-1-g_1}}{\partial I_{2g-3-k}}\label{eqn:FgNstart}\\
& +\frac{2(2g-3)!}{(2g)!(1-I_1)^2}\sum_{g_1=1}^{g-2}\frac{\partial F^N_{g_1}}{\partial I_{2g-4}} d_{X}\left(F^N_{g-1-g_1}\right) +\sum_{k=2}^{2g-4}\frac{k!(2g-2-k)!}{(2g)!(1-I_1)}\frac{\partial^2 F^N_{g-2}}{\partial I_{k-1}\partial I_{2g-3-k}},\nonumber\\
\frac{\partial F^N_{g}}{\partial I_{2g-2}}=&\frac{2N(2g-3)!}{(2g-1)!(1-I_1)}\frac{\partial F^N_{g-1}}{\partial I_{2g-4}}+\frac{(2g)!}{2!(2g-1)!}\frac{I_2}{1-I_1}\frac{\partial F^N_{g}}{\partial I_{2g-1}}\\
&+\sum_{g_1=1}^{g-2}\sum_{k=2}^{2g-5}\frac{k!(2g-3-k)!}{(2g-1)!(1-I_1)}\frac{\partial F^N_{g_1}}{\partial I_{k-1}}\frac{\partial F^N_{g-1-g_1}}{\partial I_{2g-4-k}} +\frac{2}{(1-I_1)^2}\frac{(2g-4)!}{(2g-1)!}\frac{\partial F^N_{g-2}}{\partial I_{2g-5}} d_{X}\left(F^N_{1}\right)\nonumber\\
& +\sum_{k=2}^{2g-5}\frac{k!(2g-3-k)!}{(2g-1)!(1-I_1)}\frac{\partial^2 F^N_{g-2}}{\partial I_{k-1}\partial I_{2g-4-k}} +\frac{1}{(1-I_1)^2}\frac{(2g-4)!}{(2g-1)!} \left(\frac{\partial F^N_{g-2}}{\partial I_{2g-6}}+2d_{X}\left(\frac{\partial F^N_{g-2}}{\partial I_{2g-5}}\right) \right),\nonumber\\
\vdots& \nonumber\\
\frac{\partial F^N_{g}}{\partial I_{5}}=&\frac{N}{15(1-I_1)}\frac{\partial F^N_{g-1}}{\partial I_{3}}+\sum_{n=2}^{2g-5}\frac{(n+5)!}{6!n!}\frac{I_n}{1-I_1}\frac{\partial F^N_{g}}{\partial I_{n+4}}\\
&+\sum_{g_1=1}^{g-2}\frac{1}{180(1-I_1)}\frac{\partial F^N_{g_1}}{\partial I_{1}}\frac{\partial F^N_{g-1-g_1}}{\partial I_{1}} +\frac{1}{60(1-I_1)^2}\sum_{g_1=1}^{g-2}\frac{\partial F^N_{g_1}}{\partial I_{2}} d_{X}\left(F^N_{g-1-g_1}\right)\nonumber\\
& +\frac{1}{180(1-I_1)}\frac{\partial^2 F^N_{g-2}}{\partial I_{1}\partial I_{1}} +\frac{1}{120(1-I_1)^2} \left(\frac{\partial F^N_{g-2}}{\partial I_{1}}+2d_{X}\left(\frac{\partial F^N_{g-2}}{\partial I_{2}}\right) \right).\nonumber
\end{align}
For $\frac{\partial F^{1D}_{g}}{\partial I_{i}}, (i=4,3,2,1)$,
we have by \eqref{d}-\eqref{e}:
\begin{align}
\frac{\partial F^N_{g}}{\partial I_{4}} =&\frac{N}{10(1-I_1)}\frac{\partial F^N_{g-1}}{\partial I_{2}}+\sum_{n=2}^{2g-4}\frac{(n+4)!}{5!n!}\frac{I_n}{1-I_1}\frac{\partial F^N_{g}}{\partial I_{n+3}}+\frac{1}{30(1-I_1)^2}\sum_{g_1=1}^{g-2}\frac{\partial F^N_{g_1}}{\partial I_{1}} d_{X}\left(F^N_{g-1-g_1}\right)\\
&  +\frac{1}{60(1-I_1)^3} \left(d_{X}\left(F^N_{g-2}\right)+2(1-I_1)d_{X}\left(\frac{\partial F^N_{g-2}}{\partial I_{1}}\right) \right),\nonumber\\
\frac{\partial F^N_{g}}{\partial I_{3}}=& \frac{N}{6(1-I_1)}\frac{\partial F^N_{g-1}}{\partial I_1} +\sum_{n=2}^{2g-3}\frac{(n+3)!}{4!n!}\frac{I_n}{1-I_1}\frac{\partial F^N_{g}}{\partial I_{2+n}}+\frac{1}{4!(1-I_1)^3}\sum_{g_1=1}^{g-2}d_{X}\left(F^N_{g_1}\right) d_{X}\left(F^N_{g-1-g_1}\right)  \\
&+\frac{1}{4!(1-I_1)}\left(\frac{1}{1-I_1}d_{X}\right)^2\left(F_{g-2}^N\right) +\delta_{g,2}\frac{N}{4!(1-I_1)^2},\nonumber\\
\frac{\partial F^N_{g}}{\partial I_{2}}=& \frac{N}{3(1-I_1)}\frac{1}{1-I_1} d_{X}\left(F^N_{g-1}\right) +\sum_{n=2}^{2g-2} \frac{(n+2)!}{3!n!}\frac{I_n}{1-I_1}\frac{\partial F^N_{g}}{\partial I_{n+1}},\\
\frac{\partial F^{N}_{g}}{\partial I_1}=&\frac{1}{2(1-I_1)}\sum_{n=2}^{2g-1}(n+1)I_n\frac{\partial F^{N}_{g}}{\partial I_n}.\label{eqn:FgNend}
\end{align}
By equations \eqref{eqn:FgNstart}-\eqref{eqn:FgNend},
we can solve $\left\{\frac{\pd F^{N}_{g}}{\partial I_k}\right\}_{g\geqslant2,k=2g-1,2g-2,\cdots,1}$ recursively, given the computation for $F^{N}_{k<g}$.
Since $F_{g}^{N}$ is weighted homogeneous of degree $2g-2$ in $I_k$ with $\deg I_k = k-1$,
we have
\begin{equation}
F^{N}_g=\frac{1}{2g-2}\sum_{k=1}^{2g-1}(k-1)I_k\frac{\partial F^{N}_{g}}{\partial I_k}.
\end{equation}

\begin{Example}
Let us compute $F^{N}_{2}$ by the above procedure.
We have
\begin{align}
\frac{\partial F^N_{2}}{\partial I_{3}}=& \frac{N}{6(1-I_1)}\frac{\partial F^N_{1}}{\partial I_1} +\frac{N}{4!(1-I_1)^2}=\frac{N+2N^3}{24(1-I_1)^2},\\
\frac{\partial F^N_{2}}{\partial I_{2}}=& \frac{N}{3(1-I_1)}\frac{1}{1-I_1} d_{X}\left(F^N_{1}\right) + 2\frac{I_2}{1-I_1}\frac{\partial F^N_{2}}{\partial I_{3}}=\frac{(N+4N^3)I_2}{12(1-I_1)^3},\\
\frac{\partial F^{N}_{2}}{\partial I_1}=&\frac{1}{2(1-I_1)}\sum_{l=2}^{3}(l+1)I_l\frac{\partial F^{N}_{2}}{\partial I_l}=\frac{(N+4N^3)I_2^2}{8(1-I_1)^4}+\frac{(N+2N^3)I_3}{12(1-I_1)^3},
\end{align}
therefore
\begin{align}
F_{2}^{N}=\frac{1}{2}\left(I_2\frac{\partial F^N_{2}}{\partial I_{2}}+2I_3\frac{\partial F^N_{2}}{\partial I_{3}}\right) =\frac{(N+4N^3)I_2^2}{24(1-I_1)^3}+\frac{(N+2N^3)I_3}{24(1-I_1)^2}.
\end{align}
Similarly, we have with the help of a computer program:
\begin{align}
F_{3}^N=&\left(216N^4+189N^2\right)\hat{I}_2^4+\left(216N^4+234N^2\right)\hat{I}_2^2\hat{I}_3 +\left(18N^4+30N^2\right)\hat{I}_3^2 \\
&+\left(45N^4+60N^2\right)\hat{I}_2\hat{I}_4+\left(5N^4+10N^2\right)\hat{I}_5 ,\nonumber\\
F_{4}^N=&\left(13608N^5+26892N^3+\frac{8505}{2}N\right)\hat{I}_2^6+ \left(22032 N^5+49248 N^3+8505N\right)\hat{I}_2^4\hat{I}_3\\
&+\left(7776N^5+20304N^3+3960 N\right)\hat{I}_2^2\hat{I}_3^2+ \left(288N^5+1056 N^3+240N\right)\hat{I}_3^3\nonumber\\
&+\left(5400N^5+13770 N^3+2565 N\right)\hat{I}_2^3\hat{I}_4+ \left(2160N^5+6480  N^3+1440N\right)\hat{I}_2\hat{I}_3\hat{I}_4\nonumber\\
&+\left(90N^5+300 N^3+\frac{165}{2} N\right)\hat{I}_4^2+ \left(1080N^5+3330N^3+675N\right)\hat{I}_2^2\hat{I}_5\nonumber\\
&+\left(144N^5+600 N^3+156N\right)\hat{I}_3\hat{I}_5+ \left(168 N^5+630N^3+147N\right)
\hat{I}_2\hat{I}_6\nonumber\\
&+\left(14N^5+70N^3+21N\right)\hat{I}_7,\nonumber
\end{align}
where
\be
\hat{I}_k:=\frac{1}{(k+1)!}\frac{I_k}{(1-I_1)^{(k+1)/2}}
\ee
\end{Example}

\subsection{Special deformation in thin genus expansion of the Hermitian one-matrix models in $I$-coordinates}
In \cite{Zhou4}, the second named author studied the special deformation of the Hermitian one-matrix models which is defined by:
\begin{equation}
y^{N}=\frac{1}{\sqrt{2}}\sum_{n\geqslant0}\frac{t_n-\delta_{n,1}}{n!}z^n+\frac{\sqrt{2}N}{z} +\sqrt{2}\sum_{n\geqslant1}\frac{n!}{z^{n+1}}\frac{\partial F_0^{N}}{\partial t_{n-1}}.
\end{equation}
Now we rewrite it in $I$-coordinates, we have:
\begin{Theorem}\label{thm:thinspecialdeformationN}
In $I$-coordinates, the special deformation of the Hermitian one-matrix models is as follows:
\begin{align}
y^{N}=&\frac{\sqrt{2}N}{z-I_0}+\frac{1}{\sqrt{2}}\sum_{n\geqslant1}\frac{I_n-\delta_{n,1}}{n!}(z-I_0)^n.
\end{align}
\end{Theorem}
\begin{proof}
This is just \cite[Theorem 2.1]{Zhou4}.
\end{proof}
\subsection{Special deformation in fat genus expansion of the Hermitian one-matrix models in $I$-coordinates}
In \cite{Zhou4}, the second named author also studied another special deformation of the Hermitian one-matrix models based on
the fat genus expansion, which is defined by
\be\label{fatspecialdeformation}
y^{t}=\frac{1}{\sqrt{2}}\sum_{n\geqslant0}\frac{t_n-\delta_{n,1}}{n!}z^n+\frac{\sqrt{2}t}{z} +\sqrt{2}\sum_{n\geqslant1}\frac{n!}{z^{n+1}}\frac{\partial F_0^{t}}{\partial t_{n-1}}.
\ee
We can also rewrite this special deformation in $I$-coordinates, and we get:
\begin{Theorem}\label{thm:fatspecialdeformationN}
In $I$-coordinates, the special deformation \eqref{fatspecialdeformation} can be written as:
\begin{align}
y^{t}=&\frac{\sqrt{2}t}{z-I_0}+\frac{1}{\sqrt{2}}\sum_{n\geqslant1}\frac{I_n-\delta_{n,1}}{n!}(z-I_0)^n +\frac{1}{(z-I_0)^2}
\frac{1}{1-I_1} d_X(\tilde{F}^{t}_{0}) +\sum_{l\geqslant1}\frac{(l+1)!}{(z-I_0)^{l+2}}\frac{\partial \tilde{F}^{t}_{0}}{\partial I_l}.
\end{align}
Here $\tilde{F}^{t}_{0}$ is defined as follows:
\be
\tilde{F^{t}_{0}}=F^{t}_{0}-F^{t}_{0,0},
\ee
where $F_{0}^t$ is free energy of genus $0$ in fat genus expansion and $F^{t}_{0,0}$ is defined as follows:
\be
F^{t}_{0,0}:=t\int t_0 d I_0=tF^{1D}_0.
\ee
\end{Theorem}
\begin{proof}
Similar to the proof of Theorem \ref{thm:specialdeformation1D}.
\end{proof}
In fact, in \cite[Section 7]{Zhou3}, the second named author has pointed out that one can get $\tilde{F}^t_0$ from the thin genus expansion, and one has the following result:
\begin{Proposition}\label{Prop:fatg=0}
$\tilde{F^{t}_{0}}$ has an expansion in 't Hooft coupling constant $t$:
\be
\tilde{F^{t}_{0}}=\sum_{k\geqslant1}t^{k+1}F^{t}_{0,k},
\ee
with $F^{t}_{0,k}$ weighted homogeneous of degree $2k-2$,
and it satisfies the following equations:
\begin{align}
\frac{\partial F^{t}_{0,k}}{\partial I_1}=&\frac{1}{2(1-I_1)}\sum_{n=2}^{2g-1}(n+1)I_n\frac{\partial F^{t}_{0,k}}{\partial I_n}+\delta_{k,1}\frac{1}{2(1-I_1)},\\
\frac{\partial F^t_{0,k}}{\partial I_{2}}=& \frac{1}{3(1-I_1)}\frac{1}{1-I_1} d_{X}\left(F^t_{0,k-1}\right) +\sum_{n=2}^{2k-2} \frac{(n+2)!}{3!n!}\frac{I_n}{1-I_1}\frac{\partial F^t_{0,k}}{\partial I_{n+1}},\\
\frac{\partial F^t_{0,k}}{\partial I_{3}}=& \frac{1}{6(1-I_1)}\frac{\partial F^t_{0,k-1}}{\partial I_1} +\sum_{n=2}^{2k-3}\frac{(n+3)!}{4!n!}\frac{I_n}{1-I_1}\frac{\partial F^t_{0,k}}{\partial I_{2+n}}+\frac{1}{4!(1-I_1)^3}\sum_{k_1=1}^{k-2}d_{X}\left(F^t_{0,k_1}\right) d_{X}
\left(F^t_{0,k-1-k_1}\right),  \\
\frac{\partial F^t_{0,k}}{\partial I_{4}} =&\frac{1}{10(1-I_1)}\frac{\partial F^t_{0,k-1}}{\partial I_{2}}+\sum_{n=2}^{2k-4}\frac{(n+4)!}{5!n!}\frac{I_n}{1-I_1}\frac{\partial F^t_{0,k}}{\partial I_{n+3}}+\frac{1}{30(1-I_1)^2}\sum_{k_1=1}^{g-2}\frac{\partial F^t_{0,k_1}}{\partial I_{1}} d_{X}\left(F^t_{0,k-1-k_1}\right),\\
\frac{\partial F^t_{0,k}}{\partial I_{p+4}}=&\frac{2(p+3)!}{(p+5)!(1-I_1)}\frac{\partial F^t_{0,k-1}}{\partial I_{p+2}}+\sum_{n=2}^{2k-4-p}\frac{(p+n+4)!}{n!(p+5)!}\frac{I_n}{1-I_1}\frac{\partial F^t_{0,k}}{\partial I_{p+n+3}}\nonumber\\
&+\sum_{k_1=1}^{k-2}\sum_{i=2}^{p+1}\frac{i!(p+3-i)!}{(p+5)!(1-I_1)}\frac{\partial F^t_{0,k_1}}{\partial I_{i-1}}\frac{\partial F^t_{0,k-1-k_1}}{\partial I_{p+2-i}} +\frac{2}{(1-I_1)^2}\frac{(p+2)!}{(p+5)!}\sum_{k_1=1}^{k-2}\frac{\partial F^t_{0,k_1}}{\partial I_{p+1}} d_{X}\left(F^t_{0,k-1-k_1}\right).
\end{align}
\end{Proposition}
This Proposition gives a recursively way to compute $F^{t}_{0,g}$,
which is similar to the computation of $F^{N}_{g}$. We have:
\begin{align}
F^{t}_{0,1}=&\frac{1}{2}\log\frac{1}{1-I_1},\\
F^{t}_{0,2}=&\frac{1}{6}\frac{I_2^2}{(1-I_1)^3}+\frac{1}{12}\frac{I_3}{(1-I_1)^{2}},\\
F^{t}_{0,3}=&\frac{1}{6}\frac{I_2^4}{(1-I_1)^6}+\frac{1}{4}\frac{I_2^2I_3}{(1-I_1)^{6}} +\frac{1}{32}\frac{I_3^2}{(1-I_1)^4} +\frac{1}{16}\frac{I_2I_4}{(1-I_1)^4}+\frac{1}{144}\frac{I_5}{(1-I_1)^3},\\
F^{t}_{0,4}=&\frac{7}{24}\frac{I_2^6}{(1-I_1)^9}+\frac{17}{24}\frac{I_2^4I_3}{(1-I_1)^{8}} +\frac{3}{8}\frac{I_2^2I_3^2}{(1-I_1)^7}+\frac{1}{48}\frac{I_3^3}{(1-I_1)^6} +\frac{5}{24}\frac{I_2^3I_4}{(1-I_1)^7} \nonumber+\frac{1}{8}\frac{I_2I_3I_4}{(1-I_1)^6}\\
&+\frac{1}{160}\frac{I_4^2}{(1-I_1)^{5}} +\frac{1}{24}\frac{I_2^2I_5}{(1-I_1)^6}+\frac{1}{120}\frac{I_3I_5}{(1-I_1)^5} +\frac{1}{180}\frac{I_2I_6}{(1-I_1)^5} +\frac{1}{2880}\frac{I_7}{(1-I_1)^4}.
\end{align}
In general,
\begin{align}
F^{t}_{0,k}=a_1\frac{I_{2k-1}}{(1-I_1)^k}+a_2\frac{I_2I_{2k-2}}{(1-I_1)^{k+1}} +a_3\frac{I_3I_{2k-3}}{(1-I_1)^{k+1}}+\cdots.
\end{align}

\begin{align}
a_1=&\frac{1}{k!(k+1)!},\\
a_2=&\frac{k^2}{k!(k+1)!},\\
a_3=&\frac{g^2(k-1)}{2\cdot k!(k+1)!}.
\end{align}
\begin{Remark}
One can also get the Proposition \ref{Prop:fatg=0} by Virasoro constraints for fat genus expansion as
follows.
By Proposition \ref{Prop:RenormalHMM}, one can rewrite the partition function $Z^N$ as
\begin{equation}
Z^N=e^{\frac{1}{g_s}F^{N}_0}\tilde{Z}^N,
\end{equation}
with $\tilde{Z}^N$ does not depend on $I_0$. Then $\tilde{Z}^N$ can be viewed as partition function
of Hermitian one-matrix model with $g_1=0$ and $g_k=\frac{I_{k-1}}{(k-1)!}$. The free energy function
$\tilde{F}^N=\log \tilde{Z}^N$ has a genus expansion (the fat genus expansion)
\be
\tilde{F}^N=\sum_{g\geqslant0}g_s^{2g-2}\tilde{F}_g^t,
\ee
then the two definitions of $\tilde{F}_0^t$ are coincident.
By the proof of Theorem \ref{HMM}, $\tilde{Z}^N$ satisfies the following constraints:
\begin{equation}
\tilde{L}^t_{m}\tilde{Z}^N=0,\ \ \ m\geqslant0
\end{equation}
where
\begin{align}
\tilde{L}^{t}_0=&-2\frac{\partial}{\partial I_1}+\sum_{l\geqslant 1}(l+1)I_l\frac{\partial}{\partial I_l}+\frac{t^2}{g_s^2}, \\
\tilde{L}_{1}^t=&\frac{2t}{1-I_1}d_{X}+
\sum_{n\geqslant 1}\frac{(n+2)!}{n!}(I_n-\delta_{n,1})
\frac{\partial }{\partial I_{n+1}} ,\\
\tilde{L}_{2}^t=&4t\frac{\partial }{\partial I_{1}}+
\sum_{n\geqslant 1}\frac{(n+3)!}{n!}(I_n-\delta_{n,1})
\frac{\partial }{\partial I_{n+2}}
+g_s^2\left(\frac{1}{1-I_1}d_{X}\right)^2+\frac{t}{1-I_1},\\
\tilde{L}_{m}^t=&2tm!\frac{\partial }{\partial I_{m-1}}+
\sum_{n\geqslant 1}\frac{(m+n+1)!}{n!}(I_n-\delta_{n,1})
\frac{\partial }{\partial I_{m+n}}\\
&+g_s^2(m-1)!\left(\frac{1}{1-I_1}d_{X}\frac{\partial }{\partial I_{m-2}}+
\frac{\partial }{\partial I_{m-2}}\frac{1}{1-I_1}d_{X}\right)  \nonumber\\
& +g_s^2\sum_{k=2}^{m-2}k!(m-k)!\frac{\partial }{\partial I_{k-1}}
\frac{\partial }{\partial I_{m-k-1}}+\delta_{m,2}\frac{t}{1-I_1},\nonumber
\end{align}
for $m\geqslant3$.
This gives another proof of the Proposition \ref{Prop:fatg=0}.
\end{Remark}

\section{Itzykson-Zuber Ansatz in  2D Topological Gravity}

\label{sec:IZ}

In this Section we will prove the validity of Itzykson-Zuber Ansatz in 2D topological gravity.

\subsection{Preliminary results of the 2D topological gravity}

According to Witten \cite{Wit1},
the mathematical theory of the 2D topological gravity studies the following intersection numbers on the Deligne-Mumford moduli spaces:
\begin{equation}
\left<\tau_{d_1}\cdots\tau_{d_{n}}\right>^{2D}_g:=\int_{\overline{\mathcal{M}}_{g,n}}\psi_1^{d_1}\wedge\cdots \wedge \psi_n^{d_{n}}.
\end{equation}
The free energy of the 2D topological gravity is the generating series of these intersection numbers:
\begin{equation}
F^{2D}(t)=\sum_{g=0}^{\infty}F^{2D}_g:=\sum_{g=0}^{\infty}\sum_{n_0,n_1,n_2\cdots} \left<\tau_0^{n_0}\tau_{1}^{n_{1}}\tau_{2}^{n_{2}}\cdots \right>^{2D}_g \frac{t_{0}^{n_0}}{n_0!}\frac{t_{1}^{n_1}}{n_1!}\frac{t_{2}^{n_2}}{n_2!}\cdots.
\end{equation}
It is well known that intersection number $\left<\tau_0^{n_0}\tau_{1}^{n_{1}}\tau_{2}^{n_{2}}\cdots \right>^{2D}_g$ satisfies the following selection rule:
\begin{equation}
\sum_{i}(i-1)n_i=3g-3.
\end{equation}
In other words, if we define
\begin{equation}
\deg{t_i}=i-1.
\end{equation}
then $F^{2D}_{g}$ is weighted homogeneous of degree $3g-3$.
The partition function of the topological 2D gravity is defined by:
\begin{equation}
Z^{2D}:=e^{F^{2D}}.
\end{equation}

\subsection{Virasoro constraints for  topological 2D gravity}
The following theorem is the famous Witten Conjecture/Kontsevich Theorem
\cite{Wit1, Kon}.
See also  \cite{DVV1} or \cite{IZ}.

\begin{Theorem}
The partition function $Z^{2D}$ of the topological 2D gravity satisfies the following equations:
\begin{equation}
L^{2D}_{n}Z^{2D}=0,\ \ \ n\geqslant-1
\end{equation}
where
\begin{align}
L^{2D}_{-1}&=\frac{t_0^2}{2}+\sum_{n\geqslant 1}(t_n-\delta_{n,1})\frac{\partial}{\partial t_{n-1}},\\
L^{2D}_{0}&=\frac{1}{8}+\sum_{n\geqslant 0}(2n+1)(t_n-\delta_{n,1})\frac{\partial}{\partial t_{n}},\\
L^{2D}_{m}&=\sum_{n\geqslant 0}\frac{(2n+2m+1)!!}{(2n-1)!!}(t_n-\delta_{n,1})\frac{\partial}{\partial t_{m+n}} +\frac{1}{2}\sum_{k+l=m-1}(2k+1)!!(2l+1)!!\frac{\partial^2}{\partial t_k\partial t_l},
\end{align}
for $m\geqslant1$.Furthermore, $\{L^{2D}_m\}_{m\geqslant-1}$ satisfies the following commutation relations:
\begin{equation}
\left[L^{2D}_m,L^{2D}_n\right]=(m-n)L^{2D}_{m+n},
\end{equation}
for $m,n\geqslant-1$.
\end{Theorem}

\subsection{The Itzykson-Zuber Ansatz}
It was announced in \cite{Zhou1} that the first two Virasoro constraints $L^{2D}_{-1}$ and $L^{2D}_{0}$
in $I$-coordinates can be used to solve free energies in genus $g=0,1$ and establish the Itzykson-Zuber ansatz.
Here we present some details which are similar to the 1D topological gravity case
in \cite{Zhou1}.

\begin{Theorem} \label{IZ}(Itzykson-Zuber Ansatz \cite{IZ})
For free energy of 2D gravity, we have:
\begin{align}
F^{2D}_0=&\frac{1}{6}I_0^3-\sum_{n\geqslant0}\frac{(-1)^nI_0^{n+2}}{(n+2)!} I_n
+\frac{1}{2}\sum_{j,k\geqslant0}\frac{(-1)^{j+k}I_0^{j+k+1}}{j!k!(j+k+1)}I_jI_{k},\label{eqn:F02D}\\
F^{2D}_{1}=&\frac{1}{24}\log\frac{1}{1-I_1}.\label{eqn:F12D}
\end{align}
and  for $g \geq 2$,
\begin{equation}\label{eqn:Fg2D}
F^{2D}_g(t)=\sum_{\sum_{2\leqslant k\leqslant 3g-2}(k-1)l_k=3g-3}\left<\tau_2^{l_2}\tau_{3}^{l_{3}}\cdots\tau_{3g-2}^{l_{3g-2}} \right>^{2D}_g \prod_{j=2}^{3g-2}\frac{1}{l_j!}\left(\frac{I_j}{(1-I_1)^{\frac{2j+1}{3}}}\right)^{l_j}.
\end{equation}
\end{Theorem}
\begin{proof}
We first rewrite $L_{-1}^{2D}$ in I-coordinates:
\begin{equation}
L^{2D}_{-1}=-\frac{\partial}{\partial I_0}+\frac{1}{2}\left(\sum_{n=0}^{\infty}\frac{(-1)^nI_0^n}{n!}I_n\right)^2.
\end{equation}
This gives
\begin{equation}
\frac{\partial F^{2D}}{\partial I_0}=\frac{1}{2}\left(\sum_{n=0}^{\infty}\frac{(-1)^nI_0^n}{n!}I_n\right)^2.
\end{equation}
By analyzing degree of $F^{2D}_g$, one has
\begin{align}
\frac{\partial F^{2D}_0}{\partial I_0}&=\frac{1}{2}\left(\sum_{n=0}^{\infty}\frac{(-1)^nI_0^n}{n!}I_n\right)^2,\label{2DF0}\\
\frac{\partial F^{2D}_g}{\partial I_0}&=0,\ \ \ g>0. \label{2DFg-I0}
\end{align}
This gives
\begin{align}
F^{2D}_0=&\frac{1}{6}I_0^3-\sum_{n\geqslant0}\frac{(-1)^nI_0^{n+2}}{(n+2)!} I_n +\frac{1}{2}\sum_{n,k\geqslant0}\frac{(-1)^{n+k}I_0^{n+k+1}}{n!k!(n+k+1)}I_nI_{k},
\end{align}
and $F^{2D}_g$ is independent on $I_0$.
By
$$L^{2D}_{0}Z^{2D}=0,$$
one has
\begin{equation}
0=\frac{1}{8}+\sum_{n\geqslant 0}(2n+1)(t_n-\delta_{n,1})\frac{\partial F^{2D}}{\partial t_{n}}.
\end{equation}
Let $I_0=0$, we have
\begin{equation}
0=\frac{1}{8}+\sum_{n\geqslant 1}(2n+1)(I_n-\delta_{n,1})\frac{\partial F^{2D}}{\partial I_{n}}.
\end{equation}
This is equal to
\begin{align}
0=&\sum_{n\geqslant 1}(2n+1)(I_n-\delta_{n,1})\frac{\partial F_1^{2D}}{\partial I_{n}}+\frac{1}{8}\delta_{g,1}.\label{L02DinI}
\end{align}
Write $F^{2D}_g$ as formal power series in $I_1$, with coefficients a priori formal series in $I_2,I_3,\cdots$:
\begin{equation}
F_{g}^{2D}=\sum_{n=0}^{\infty}a_{g,n}(I_2,I_3,\cdots)I_1^n.
\end{equation}
Write
\begin{equation}
a_{g,0}=\sum \alpha_{l_2,l_3,\cdots,l_m}\prod_{i=2}^m\frac{I_i^{l_i}}{l_i!}.
\end{equation}
Since $I_0=0$ is equivalent to $t_0=0$ and $ t_k=I_k (k>0)$,
\begin{equation}
\alpha_{l_2,l_3,\cdots,l_m}=\left<\tau_2^{l_2}\tau_{3}^{l_{3}}\cdots\tau_{m}^{l_{m}} \right>^{2D}_g.
\end{equation}
This vanishes unless the following selection rule is satisfied:
\begin{equation}
\sum_{i=2}^{m}(i-1)l_i=3g-3.
\end{equation}
So $l_i=0$ unless $i\leqslant3g-2$. Therefore
\begin{equation}
a_{g,0}=\sum \left<\tau_2^{l_2}\tau_{3}^{l_{3}}\cdots\tau_{3g-2}^{l_{3g-2}} \right>^{2D}_g \prod_{i=2}^{3g-2}\frac{I_i^{l_i}}{l_i!},
\end{equation}
where the summation is taken over all nonnegative integers $l_2,l_3,\cdots,l_{3g-2}$ such that
\begin{equation}
\sum_{i=2}^{3g-2}(i-1)l_i=3g-3.
\end{equation}
Let
\begin{equation}
E=\sum_{k\geqslant2}\frac{2k+1}{3}I_k\frac{\partial }{\partial I_k}.
\end{equation}
Then the equation \eqref{L02DinI} gives us the following recursion relations:
\begin{align}
a_{g,1}=&E(a_{g,0})+\frac{1}{24}\delta_{g,1},\\
ma_{g,m}=&(m-1)a_{g,m-1}+E(a_{g,m-1}).
\end{align}
When $g=1$, we have $a_{1,0}=0$ by selection rule. One can see that $a_{1,n}=\frac{1}{24n}$ for $n\geqslant1$. Therefore
\begin{equation}
F_{1}^{2D}=\sum_{n=1}^{\infty}\frac{1}{24n}I_1^n=\frac{1}{24}\log{\frac{1}{1-I_1}}.
\end{equation}
When $g>1$, one finds
\begin{equation}
\begin{split}
a_{g,n}=&\sum_{\sum_{j=2}^{3g-2}(j-1)l_j=3g-3} \left<\tau_2^{l_2}\tau_{3}^{l_{3}}\cdots\tau_{3g-2}^{l_{3g-2}} \right>^{2D}_g (-1)^n\binom{g-1-\sum_{j=2}^{3g-2}jl_j}{n} \prod_{i=2}^{3g-2}\frac{I_i^{l_i}}{l_i!}.
\end{split}
\end{equation}
This proves
\begin{align}
F_{g}^{2D}=&\sum_{\sum_{j=2}^{3g-2}(j-1)l_j=3g-3} \left<\tau_2^{l_2}\tau_{3}^{l_{3}}\cdots\tau_{3g-2}^{l_{3g-2}} \right>^{2D}_g \frac{1}{(1-I_1)^{\sum_{j=2}^{3g-2}jl_j-g+1}} \prod_{i=2}^{3g-2}\frac{I_i^{l_i}}{l_i!}\\
=&\sum_{\sum_{j=2}^{3g-2}(j-1)l_j=3g-3} \left<\tau_2^{l_2}\tau_{3}^{l_{3}}\cdots\tau_{3g-2}^{l_{3g-2}} \right>^{2D}_g \prod_{i=2}^{3g-2}\frac{1}{l_i!}\left(\frac{I_j}{(1-I_1)^{(2j+1)/2}}\right)^{l_j}.\nonumber
\end{align}
\end{proof}

\begin{Remark}
The formula for $F^{2D}_0$ in \eqref{eqn:F02D} is equivalent to the version announced in \cite{Zhou1}.
It is different from the original formula given by Itzyskon-Zuber \cite{IZ}.
Their version will be proved in Corollary \ref{cor:IZ}.
\end{Remark}

In \cite{Zhou1}, the second author studied another kind of coordinates $\{\frac{\partial^n I_0}{\partial t_0^n}\}_{n\geqslant 0}$, and proved the following transformation equations between this coordinates and $I$-coordinates:
\begin{equation}\label{eqn:Itojet}
\frac{\partial^n I_0}{\partial t_0^n}=\sum_{\sum_{j\geqslant1}jm_j=n-1} \frac{(\sum_j(j+1)m_j)!}{\prod_j((j+1)!)^{m_j}m_j!}\cdot \frac{\prod_jI_{j+1}^{m_j}}{(1-I_1)^{\sum_{j}(j+1)m_j+1}},
\end{equation}

\begin{equation}\label{eqn:jettoI}
I_n=-\sum_{\sum_{j\geqslant1}jm_j=n-1} \frac{(\sum_j(j+1)m_j)!}{\prod_j((j+1)!)^{m_j}m_j!}\cdot \frac{\prod_j\left(-\frac{\partial^{j+1} I_0}{\partial t_0^{j+1}}\right)^{m_j}}{\left(\frac{\partial I_0}{\partial t_0}\right)^{\sum_{j}(j+1)m_j+1}}.
\end{equation}

In an earlier work by Euguchi, Yamada and Yang \cite{Egu1},
the Itzykson-Zuber Ansatz is written in two forms:
\begin{align}
F_{g}^{2D}=&\sum_{\sum_{j=2}^{3g-2}(j-1)l_j=3g-3} a_{l_2\cdots l_{3g-2}}\frac{(u'')^{l_2}\cdots(u^{(3g-2)})^{l_{3g-2}}}{(u')^{2(1-g)+\sum_{j=2}^{3g-2}jl_j}} \label{eqn:injet-coor}\\
=&\sum_{\sum_{j=2}^{3g-2}(j-1)l_j=3g-3} b_{l_2\cdots l_{3g-2}}\frac{I_2^{l_2}\cdots I_{3g-2}^{l_{3g-2}}}{(u')^{2(1-g)+\sum_{j=2}^{3g-2}l_j}},\label{eqn:inI-coor}
\end{align}
where $u=I_0$, and $u^{(n)} = \frac{\pd^n I_0}{\pd t_0^n} $.
They proved \eqref{eqn:injet-coor} by KdV hierarchy plus string equation,
and derived \eqref{eqn:inI-coor} by \eqref{eqn:injet-coor}
following a proposal by Itzyson and Zuber \cite{IZ}.
In this Subsection, we have proved \eqref{eqn:inI-coor} by Virasoro constraints,
and by \eqref{eqn:jettoI}, one can derives \eqref{eqn:injet-coor} by \eqref{eqn:inI-coor}.
Our proof seems to be simpler because we use only two linear equations while
the KdV hierarchy is a sequence of nonlinear equations.
Furthermore, in our proof,
we have only used $L^{2D}_{-1}$ and $L^{2D}_{0}$ in $I$-coordinates.
In the nex Subsection,
we will show that it is possible to use the higher Virasoro constraints in $I$-coordinates
to derive a recursive method to solve free energies in higher genera.

\subsection{Computations of $F^{2D}_g$ by Virasoro constraints in $I$-coordinates}

 In the above we have shown that for $g \geq 2$,
$F_g$ is independent of $I_0$.
It follows that after we rewrite the Virasoro constraints in $I$-coordinates,
we can take $I_0 =0$ to compute $F^{2D}_g$.

\begin{Theorem}\label{2Dgravity}
For free energy of the topological 2D gravity of $g\geqslant2$, the following equations hold:
\begin{align}
\frac{\partial F^{2D}_g}{\partial J_{1}}=&\sum_{p=2}^{3g-2}(2p+1)\frac{J_{p}}{1-3J_1}\frac{\partial F^{2D}_g}{\partial J_{p}},\label{ag12D}\\
\frac{\partial F^{2D}_g}{\partial J_{2}}=&\sum_{p=2}^{3g-3}(2p+1)\frac{J_{p}}{1-3J_1}\frac{\partial F^{2D}_g}{\partial J_{1+p}}
+\frac{1}{2}\frac{1}{1-3J_1}\left(\sum_{l=1}^{3g-4}(2l+3)\frac{J_{l+1}}{1-3J_1} \frac{\partial}{\partial J_l} \right)^2\left(F^{2D}_{g-1}\right)
 \label{ag22D}\\
&+\frac{1}{2}\frac{1}{1-3J_1}\sum_{g_1=1}^{g-1} \sum_{l_1=1}^{3g_1-2}(2l_1+3) \frac{J_{l_1+1}}{1-3J_1}\frac{\partial F^{2D}_{g_1}}{\partial J_{l_1}} \sum_{l_2=1}^{3g-3g_1-2}(2l_2+3)\frac{J_{l_2+1}}{1-3J_1}\frac{\partial F^{2D}_{g-g_1}}{\partial J_{l_2}},\nonumber\\
\frac{\partial F^{2D}_g}{\partial J_{3}}=&\sum_{p=2}^{3g-4}(2p+1)\frac{J_{p}}{1-3J_1}\frac{\partial F^{2D}_g}{\partial J_{2+p}}
+\sum_{g_1=1}^{g-1} \frac{\partial F^{2D}_{g_1}}{\partial J_{1}} \sum_{l=1}^{3g-3g_1-2}(2l+3)\frac{J_{l+1}}{(1-3J_1)^2}\frac{\partial F^{2D}_{g-g_1}}{\partial J_{l}}
\label{ag32D}\\
&+\sum_{l=1}^{3g-5}(2l+3)\frac{J_{l+1}}{(1-3J_1)^2} \frac{\partial^2 F^{2D}_{g-1}}{\partial J_l\partial J_1}
+3\sum_{l=1}^{3g-5}(2l+3)\frac{J_{l+1}}{(1-3J_1)^3}\frac{\partial F^{2D}_{g-1}}{\partial J_l}, \nonumber\\
\frac{\partial F^{2D}_g}{\partial J_{r+1}}=&\sum_{p=2}^{3g-2-r}(2p+1)\frac{J_{p}}{1-3J_1}\frac{\partial F^{2D}_g}{\partial J_{r+p}} +\sum_{g_1=1}^{g-1} \frac{\partial F^{2D}_{g_1}}{\partial J_{r-1}} \sum_{l=1}^{3g-3g_1-2}(2l+3)\frac{J_{l+1}}{(1-3J_1)^2}\frac{\partial F^{2D}_{g-g_1}}{\partial J_{l}} \\
&+\sum_{l=1}^{3g-3-r}(2l+3)\frac{J_{l+1}}{(1-3J_1)^2} \frac{\partial^2 F^{2D}_{g-1}}{\partial J_{l}\partial J_{r-1}}
+\frac{2r-1}{(1-3J_1)^2} \frac{\partial F^{2D}_{g-1}}{\partial J_{r-2}}\nonumber\\
&+\frac{1}{2}\frac{1}{1-3J_1}\left(\sum_{g_1=1}^{g-1}\sum_{l_1=1}^{r-2} \frac{\partial F^{2D}_{g_1}}{\partial J_{l_1}}\frac{\partial F^{2D}_{g-g_1}}{\partial J_{r-1-l_1}}
+\sum_{l_1=1}^{r-2}\frac{\partial^2 F^{2D}_{g-1}}{\partial J_{l_1}\partial J_{r-1-l_1}} \right),\nonumber
\end{align}
for $r=3,4,\cdots3g-3$, where
$$J_{k}=\frac{I_k}{(2k+1)!!}.$$
\end{Theorem}
\begin{proof}
By the constraint $L_mZ^{2D}=0$ for $m \geq 1$,
one gets:
\begin{align}
0=&\sum_{n\geqslant 0}\frac{(2n+2m+1)!!}{(2n-1)!!}(t_n-\delta_{n,1})\frac{\partial F_g^{2D}}{\partial t_{m+n}} +\frac{1}{2}\sum_{k+l=m-1}(2k+1)!!(2l+1)!!\frac{\partial^2 F_{g-1}^{2D}}{\partial t_k\partial t_l}
 \label{2DFg}\\
 &+\frac{1}{2}\sum_{g_1=0}^g\sum_{k+l=m-1}(2k+1)!!(2l+1)!!\frac{\partial^2 F_{g_1}^{2D}}{\partial t_k}\frac{\partial F_{g-g_1}^{2D}}{\partial t_l}.\nonumber
\end{align}
We want to rewrite these equations in $I$-coordinates under the condition $I_0=0$. By \eqref{eqn:F02D},
\begin{align}
\frac{\partial F_{0}^{2D}}{\partial I_{0}}=&\frac{1}{2}\left(\sum_{k\geqslant0}\frac{(-1)^kI_0^k}{k!}I_k\right)^2,\\
\frac{\partial F_{0}^{2D}}{\partial I_{n}}=&\frac{(-1)^{n+1}I_0^{n+2}}{(n+2)!} +\sum_{k\geqslant0}\frac{(-1)^{n+k}I_{0}^{n+k+1}}{n!k!(n+k+1)}I_k, \ \ \ n>0.
\end{align}
Hence
\begin{align}
\frac{\partial F_{0}^{2D}}{\partial I_{0}}\bigg|_{I_0=0}=\frac{\partial F_{0}^{2D}}{\partial I_{n}}\bigg|_{I_0=0}=0,
\end{align}
and this gives
\begin{align}
\frac{\partial F_{0}^{2D}}{\partial t_{0}}\bigg|_{I_0=0}=\frac{\partial F_{0}^{2D}}{\partial t_{n}}\bigg|_{I_0=0}=0.
\end{align}
Now we restrict equation \eqref{2DFg} in $I_0=0$, and $g\geqslant2$, we have:\\
for $m=1$:
\begin{align}
0=&\sum_{n\geqslant 1}\frac{(2n+3)!!}{(2n-1)!!}(I_n-\delta_{n,1})\frac{\partial F_g^{2D}}{\partial I_{n+1}} +\frac{1}{2}\left(\sum_{l=1}^{2g-1}\frac{I_{l+1}}{1-I_1}\frac{\partial}{\partial I_l}\right)^2(F_{g-1}^{2D}) \\
&+\frac{1}{2}\sum_{g_1=1}^{g-1}\left(\sum_{l=1}^{2g-1}\frac{I_{l+1}}{1-I_1}\frac{\partial F_{g_1}^{2D}}{\partial I_l}\right)\left(\sum_{l=1}^{2g-1}\frac{I_{l+1}}{1-I_1}\frac{\partial F_{g-g_1}^{2D}}{\partial I_l}\right),\nonumber
\end{align}
for $m=2$:
\begin{align}
0=&\sum_{n\geqslant 1}\frac{(2n+5)!!}{(2n-1)!!}(I_n-\delta_{n,1})\frac{\partial F_g^{2D}}{\partial I_{n+2}} +3!!\sum_{l\geqslant1}\frac{I_{l+1}}{1-I_1}\frac{\partial^2 F_{g-1}^{2D}}{\partial I_l\partial I_1} +\frac{3!!}{2}\sum_{l\geqslant1}\frac{I_{l+1}}{(1-I_1)^2}\frac{\partial F_{g-1}^{2D}}{\partial I_l}  \\
& +\sum_{g_1=1}^{g-1}3!!\frac{\partial^2 F_{g_1}^{2D}}{\partial I_1}\sum_{l\geqslant1} \frac{I_{l+1}}{1-I_1}\frac{\partial F_{g-g_1}^{2D}}{\partial I_l},\nonumber
\end{align}
for $m\geqslant3$:
\begin{align}
0=&\sum_{n\geqslant 1}\frac{(2n+2m+1)!!}{(2n-1)!!}(I_n-\delta_{n,1})\frac{\partial F_g^{2D}}{\partial I_{m+n}} +\frac{1}{2}\sum_{k=1}^{m-2}(2k+1)!!(2m-2k-1)!!\frac{\partial^2 F_{g-1}^{2D}}{\partial I_k\partial I_{m-k-1}} \\
&+(2m-1)!!\sum_{l\geqslant1}\frac{I_{l+1}}{1-I_1}\frac{\partial^2 F_{g-1}^{2D}}{\partial I_l\partial I_{m-1}} +\frac{(2m-1)!!}{2}\frac{\partial F_{g-1}^{2D}}{\partial I_{m-2}} \nonumber\\ &+\frac{1}{2}\sum_{g_1=1}^{g-1}\sum_{k=1}^{m-2}(2k+1)!!(2m-2k-1)!!\frac{\partial^2 F_{g_1}^{2D}}{\partial I_k}\frac{\partial F_{g-g_1}^{2D}}{\partial I_{m-k-1}} +\sum_{g_1=1}^{g-1}(2m-1)!!\frac{\partial^2 F_{g_1}^{2D}}{\partial I_{m-1}} \sum_{l\geqslant1}\frac{I_{l+1}}{1-I_1}\frac{\partial F_{g-g_1}^{2D}}{\partial I_l}.\nonumber
\end{align}
Together with Theorem \ref{IZ}, we complete the proof.
\end{proof}
We explain how to use these equations to solve $F^{2D}_{g}$ for $g\geqslant2$. By above theorem,
\begin{align}
\frac{\partial F^{2D}_g}{\partial J_{3g-2}}=&\frac{6g-7}{(1-3J_1)^2} \frac{\partial F^{2D}_{g-1}}{\partial J_{3(g-1)-2}}+\frac{1}{2}\frac{1}{1-3J_1}\left(\sum_{l_1=1}^{3g-5}\frac{\partial^2 F^{2D}_{g-1}}{\partial J_{l_1}\partial J_{3g-4-l_1}} +\sum_{g_1=1}^{g-1}\sum_{l_1=1}^{3g-5} \frac{\partial F^{2D}_{g_1}}{\partial J_{l_1}}\frac{\partial F^{2D}_{g-g_1}}{\partial J_{3g-4-l_1}}\right),\label{eqn:Fg1pt}\\
\frac{\partial F^{2D}_g}{\partial J_{3g-3}}=&\frac{5J_{2}}{1-3J_1}\frac{\partial F^{2D}_g}{\partial J_{3g-2}}+\frac{5J_{2}}{(1-3J_1)^2}\left(\frac{\partial F^{2D}_{1}}{\partial J_{1}}\frac{\partial F^{2D}_{g-1}}{\partial J_{3(g-1)-2}}
+\frac{\partial^2 F^{2D}_{g-1}}{\partial J_{1}\partial J_{3(g-1)-2}}\right)\\
&+\frac{6g-9}{(1-3J_1)^2} \frac{\partial F^{2D}_{g-1}}{\partial J_{3g-6}}
+\frac{1}{2}\frac{1}{1-3J_1}\left(\sum_{g_1=1}^{g-1}\sum_{l_1=1}^{3g-6} \frac{\partial F^{2D}_{g_1}}{\partial J_{l_1}}\frac{\partial F^{2D}_{g-g_1}}{\partial J_{3g-5-l_1}}
+\sum_{l_1=1}^{3g-6}\frac{\partial^2 F^{2D}_{g-1}}{\partial J_{l_1}\partial J_{3g-5-l_1}}\right),\nonumber \\
\vdots&\nonumber\\
\frac{\partial F^{2D}_g}{\partial J_{4}}=&\frac{5}{(1-3J_1)^2} \frac{\partial F^{2D}_{g-1}}{\partial J_{1}} +\sum_{p=2}^{3g-5}\frac{(2p+1)J_{p}}{1-3J_1}\frac{\partial F^{2D}_g}{\partial J_{p+3}} +\sum_{g_1=1}^{g-1} \frac{\partial F^{2D}_{g_1}}{\partial J_{2}} \sum_{l=1}^{3g-3g_1-2}\frac{(2l+3)J_{l+1}}{(1-3J_1)^2}\frac{\partial F^{2D}_{g-g_1}}{\partial J_{l}} \\
&+\sum_{l=1}^{3g-6}\frac{(2l+3)J_{l+1}}{(1-3J_1)^2} \frac{\partial^2 F^{2D}_{g-1}}{\partial J_{l}\partial J_{2}} +\frac{1}{2}\frac{1}{1-3J_1}\left(\sum_{g_1=1}^{g-1} \frac{\partial F^{2D}_{g_1}}{\partial J_{1}}\frac{\partial F^{2D}_{g-g_1}}{\partial J_{1}}
+\frac{\partial^2 F^{2D}_{g-1}}{\partial J_{1}^2} \right),\nonumber\\
\frac{\partial F^{2D}_g}{\partial J_{3}}=&\sum_{p=2}^{3g-4}(2p+1)\frac{J_{p}}{1-3J_1}\frac{\partial F^{2D}_g}{\partial J_{2+p}}
+\sum_{g_1=1}^{g-1} \frac{\partial F^{2D}_{g_1}}{\partial J_{1}} \sum_{l=1}^{3g-3g_1-2}(2l+3)\frac{J_{l+1}}{(1-3J_1)^2}\frac{\partial F^{2D}_{g-g_1}}{\partial J_{l}} \\
&+\sum_{l=1}^{3g-5}(2l+3)\frac{J_{l+1}}{(1-3J_1)^2} \frac{\partial^2 F^{2D}_{g-1}}{\partial J_l\partial J_1}
+3\sum_{l=1}^{3g-5}(2l+3)\frac{J_{l+1}}{(1-3J_1)^3}\frac{\partial F^{2D}_{g-1}}{\partial J_l}, \nonumber\\
\frac{\partial F^{2D}_g}{\partial J_{2}}=&\sum_{p=2}^{3g-3}(2p+1)\frac{J_{p}}{1-3J_1}\frac{\partial F^{2D}_g}{\partial J_{1+p}}
+\frac{1}{2}\frac{1}{1-3J_1}\left(\sum_{l=1}^{3g-4}(2l+3)\frac{J_{l+1}}{1-3J_1} \frac{\partial}{\partial J_l} \right)^2\left(F^{2D}_{g-1}\right)\\
&+\frac{1}{2}\frac{1}{1-3J_1}\sum_{g_1=1}^{g-1} \sum_{l_1=1}^{3g_1-2}(2l_1+3) \frac{J_{l_1+1}}{1-3J_1}\frac{\partial F^{2D}_{g_1}}{\partial J_{l_1}} \sum_{l_2=1}^{3g-3g_1-2}(2l_2+3)\frac{J_{l_2+1}}{1-3J_1}\frac{\partial F^{2D}_{g-g_1}}{\partial J_{l_2}},
\nonumber \\
\frac{\partial F^{2D}_g}{\partial J_{1}}=&\sum_{p=2}^{3g-2}(2p+1)\frac{J_{p}}{1-3J_1}\frac{\partial F^{2D}_g}{\partial J_{p}}.
\end{align}
Since $F^{2D}_{g}$ is weighted homogeneous of degree $3g-3$ with $\deg J_k=k-1$, one has:
\begin{equation}
F_{g}^{2D}=\frac{1}{3g-3}\sum_{i=1}^{3g-2}(i-1)J_{i}\frac{\partial F^{2D}_{g}}{\partial J_i}.
\end{equation}

\begin{Example}Let us compute $F^{2D}_{2}$ by the above procedure.
\begin{align}
\frac{\partial F^{2D}_2}{\partial J_{4}}=&\frac{5}{(1-3J_1)^2} \frac{\partial F^{2D}_{1}}{\partial J_{1}}+\frac{1}{2}\frac{1}{1-3J_1}\left(\frac{\partial^2 F^{2D}_{1}}{\partial J_{1}^2} +\left(\frac{\partial F^{2D}_{1}}{\partial J_{1}}\right)^2\right)=\frac{105}{128}\frac{1}{(1-3J_1)^3},\\
\frac{\partial F^{2D}_2}{\partial J_{3}}=&\frac{5J_{2}}{1-3J_1}\frac{\partial F^{2D}_2}{\partial J_{4}}
+\frac{5J_{2}}{(1-3J_1)^2}\left(\left(\frac{\partial F^{2D}_{1}}{\partial J_{1}}\right)^2 + \frac{\partial^2 F^{2D}_{1}}{\partial J_1^2} \right) +\frac{15J_{2}}{(1-3J_1)^3}\frac{\partial F^{2D}_{1}}{\partial J_1} \\
=&\frac{1015}{128}\frac{J_2}{(1-3J_1)^4},\nonumber\\
\frac{\partial F^{2D}_2}{\partial J_{2}}=&\frac{5J_{2}}{1-3J_1}\frac{\partial F^{2D}_2}{\partial J_{3}}+\frac{7J_{3}}{1-3J_1}\frac{\partial F^{2D}_2}{\partial J_{4}}
+\frac{1}{2}\frac{1}{1-3J_1}\left(\frac{5J_{2}}{1-3J_1} \frac{\partial}{\partial J_1} +\frac{7J_{3}}{1-3J_1} \frac{\partial}{\partial J_2} \right)^2\left(F^{2D}_{1}\right)\\
&+\frac{1}{2}\frac{25J_2^2}{(1-3J_1)^3}\left(\frac{\partial F^{2D}_{1}}{\partial J_1}\right)^2\nonumber\\
=&\frac{6300}{128}\frac{J_2^2}{(1-3J_1)^5}+\frac{1015}{128}\frac{J_3}{(1-3J_1)^4}.\nonumber
\end{align}
so we have
\begin{align}
F_{2}^{2D}=&\frac{1}{3}\left(\frac{105}{128}\frac{3J_4}{(1-3J_1)^3} +\frac{1015}{128}\frac{2J_2J_3}{(1-3J_1)^4}+\frac{6300}{128}\frac{J_2^3}{(1-3J_1)^5} +\frac{1015}{128}\frac{J_2J_3}{(1-3J_1)^4} \right)\\
=&\frac{2100}{128}\frac{J_2^3}{(1-3J_1)^5} +\frac{1015}{384}\frac{J_2J_3}{(1-3J_1)^4} +\frac{105}{128}\frac{J_4}{(1-3J_1)^3}.
\end{align}
This can be written in $I$-coordinates as:
\begin{align}
F^{2D}_{2}=&\frac{7}{1440}\check{I}_2^3+\frac{29}{5760}\check{I}_2\check{I}_3+\frac{1}{1152}\check{I}_4,
\end{align}
where
\begin{equation}
\check{I}_j=\frac{I_j}{(1-I_1)^{(2j+1)/3}}.
\end{equation}
Similarly, we have with the help of a computer program:
\begin{align}
F^{2D}_{3}=&\frac{245}{20736}\check{I}_2^6+\frac{193}{6912}\check{I}_2^4\check{I}_3 +\frac{205}{13824}\check{I}_2^2\check{I}_3^3 +\frac{583}{580608}\check{I}_3^3 +\frac{53}{6912}\check{I}_2^3\check{I}_4+\frac{1121}{241920}\check{I}_2\check{I}_3\check{I}_4 \\ &+\frac{607}{2903040}\check{I}_4^2+\frac{17}{11520}\check{I}_2^2\check{I}_5 +\frac{503}{1451520}\check{I}_3\check{I}_5 +\frac{77}{414720}\check{I}_2\check{I}_6+\frac{1}{82944}\check{I}_7,\nonumber\\
F^{2D}_{4}=&\frac{259553}{2488320}\check{I}_2^9+\frac{475181}{1244160}\check{I}_2^7\check{I}_3 +\frac{145693}{331776}\check{I}_2^5\check{I}_3^2 +\frac{43201}{248832}\check{I}_2^3\check{I}_3^3 +\frac{134233}{7962624}\check{I}_2\check{I}_3^4+\frac{14147}{124416}\check{I}_2^6\check{I}_4 \\ &+\frac{83851}{414720}\check{I}_2^4\check{I}_3\check{I}_4+\frac{26017}{331776}
\check{I}_2^2\check{I}_3^2\check{I}_4 +\frac{185251}{49766400}\check{I}_3^3\check{I}_4 +\frac{5609}{276480}\check{I}_2^3\check{I}_4^2 +\frac{177}{20480}\check{I}_3\check{I}_4^2 \nonumber\\ &+\frac{175}{995328}\check{I}_4^3+\frac{21329}{829440}\check{I}_2^5\check{I}_5 +\frac{13783}{414720}\check{I}_2^3\check{I}_3\check{I}_5 +\frac{1837}{259200}
\check{I}_2\check{I}_3^2\check{I}_5 +\frac{7597}{1382400}\check{I}_2^2\check{I}_4\check{I}_5\nonumber\\
&+\frac{719}{829440}\check{I}_3\check{I}_4\check{I}_5 +\frac{533}{1935360}\check{I}_2\check{I}_5^2 +\frac{2471}{552960}\check{I}_2^4\check{I}_6+\frac{7897}{2073600}\check{I}_2^2\check{I}_3\check{I}_6 +\frac{1997}{6635520}\check{I}_3^2\check{I}_6\nonumber\\
&+\frac{1081}{2322432}\check{I}_2\check{I}_4\check{I}_6 +\frac{487}{18579456}\check{I}_5\check{I}_6 +\frac{4907}{8294400}\check{I}_2^3\check{I}_7+\frac{16243}{58060800}\check{I}_2\check{I}_3\check{I}_7 +\frac{1781}{92897280}\check{I}_4\check{I}_7\nonumber\\
&+\frac{53}{921600}\check{I}_2^2\check{I}_8 +\frac{947}{92897280}\check{I}_3\check{I}_8+\frac{149}{39813120}\check{I}_2\check{I}_9 +\frac{1}{7962624}\check{I}_{10}.\nonumber
\end{align}
\end{Example}

\section{Special Deformation of the Airy Curve in Renormalized Coupling Constants}

\label{sec:Ghost}

In this Section we reformulate the special deformation of Airy curve
introduced in the setting of 2D quantum gravity in \cite{Zhou5} using renormalized coupling constants.
It is remarkable that we need to consider ghost variables introduced in \cite{Zhou6}
and consider their renormalizations.

\subsection{Special deformation of the Airy curve and ghost variables in 2D topological gravity}

In \cite{Zhou5}, the second named author studied the following special deformation of Airy curve
defined by:
\be \label{eqn:Airy}
w^{2D}:=z^{\frac{1}{2}}-\sum_{n=0}^{\infty}\frac{t_n}{(2n-1)!!}z^{n-\frac{1}{2}} -\sum_{n=0}^{\infty}(2n+1)!!\frac{\partial F_0^{2D}}{\partial t_n}z^{-n-\frac{3}{2}}.
\ee
When all $t_n$ are set to be equal to $0$, one gets a plane algebraic curve
\be
(w^{2D})^2=z.
\ee
This is called the Airy curve because its quantization gives the Airy equation.

To make sense of the right-hand side of \eqref{eqn:Airy},
the following extension of the free energy $F^{2D}_0$ was introduced in \cite{Zhou6}:
\be \label{eqn:Tilde-F02D}
\tilde{F}_0^{2D} = F_0^{2D} + \sum_{n \geq 0} (-1)^n (t_n-\delta_{n,1}) t_{-n-1},
\ee
where $t_{-1}, t_{-2}, \dots$ are formal variables referred to as the ghost variables.
The reason for adding these extra terms to $F_0^{2D}$ is that the moduli spaces
$\Mbar_{0,1}$ and $\Mbar_{0,2}$ do not make sense geometrically,
however,
from the following formula for $n \geq 3$:
\be
\sum_{m_1, \dots, m_n \geq 0} \corr{\tau_{m_1} \cdots \tau_{m_n}}_0 x_1^{m_1} \cdots x_n^{m_n}
= (x_1+ \cdots +x_n)^{n-3}.
\ee
it is customary to use the following conventions:
\ben
&& \corr{\tau_n}_0 = \delta_{n, -2}, \\
&& \corr{\tau_k\tau_{-k-1}}_0= (-1)^k.
\een
Consider the generating series
\be
\begin{split}
\sum_{n \in \bZ} \frac{\pd \tilde{F}_0}{\pd t_n} \cdot x^{n+1}
= & \sum_{n \geq 0} (-1)^n (t_{n}-\delta_{n,1}) x^{-n}
+  \sum_{n \geq 0} \biggl( \frac{\pd F_0}{\pd t_n}(\bt)  + (-1)^n t_{-n-1} \biggr) \cdot x^{n+1}.
\end{split}
\ee
and its Laplace transform:
\be
\begin{split}
& \sum_{n \in \bZ} \frac{\pd \tilde{F}_0}{\pd t_n} \cdot  \int_0^{\infty}
\frac{1}{\sqrt{x}} e^{-zx}x^{n+1} dx \\
= & \sum_{n \geq 0} (-1)^n(t_{n}-\delta_{n,1})\Gamma(-n+\frac{1}{2}) z^{n-1/2}
+  \sum_{n \geq 0} \biggl( \frac{\pd F_0}{\pd t_n}(\bt)  + (-1)^n t_{-n-1} \biggr) \cdot  \Gamma(n+\frac{3}{2}) z^{-n-3/2}\\
= & \sum_{n \in \mathbb{Z}} (t_{n}-\delta_{n,1})(-1)^n \Gamma(-n+\frac{1}{2}) z^{n-1/2}
+  \sum_{n \geq 0}\biggl( \frac{\pd F_0}{\pd t_n}(\bt)  + 2(-1)^n t_{-n-1} \biggr) \cdot  \Gamma(n+\frac{3}{2}) z^{-n-3/2}.
\end{split}
\ee
After setting $t_{-n} = 0$ for $n \geq 1$,
the right-hand side of the last equality gives us $w^{2D}$,
up to a factor of $\sqrt{\pi}$.

\subsection{Renormalized ghost variables}

From the above discussions we are lead to consider:
\be
\sum_{n\in \mathbb{Z}}(-1)^nt_n\Gamma(-n+\frac{1}{2})z^{n-\frac{1}{2}}.
\ee
This seems to play the role of $S(z)$ in 1D topological gravity.
Recall the renomralization of the coupling constants in $S(z)$ leads us to the following identity
\begin{equation}   \label{eqn:S-I}
\sum_{n=0}^{\infty}(t_{n}-\delta_{n,1})\frac{x^{n+1}}{(n+1)!}
= I_{-1}-\half I_0^2 +\sum_{n=2}^{\infty}(I_{n-1}-\delta_{n,2}) \frac{(x-I_0)^{n}}{n!},
\end{equation}
where $I_{-1}$ is defined by:
\ben
I_{-1} = \sum_{n =0}^\infty t_n \frac{I_0^{n+1}}{(n+1)!}.
\een
Also recall for $k \geq 0$,
\ben
I_k= \sum_{n \geq 0} t_{n+k} \frac{I_0^n}{n!}.
\een
Therefore,
after the introduction of the ghost variables
$\{t_{-n}\}_{n\in \mathbb{Z}_{+}}$,
we can introduce the renormalized ghost variables  $\{I_{-k}\}_{k\in \mathbb{Z}_{+}}$
as follows:
\be
I_{-k}=\sum_{n=0}^{\infty}t_{n-k} \frac{I_0^n}{n!}.
\ee
Recall for $n \geq 0$ we have \cite[Prop. 2.4]{Zhou1}:
\be \label{eqn:T-in-I}
t_n = \sum_{k=0}^\infty I_{n+k}\frac{(-1)^k I_0^k}{k!}.
\ee

\begin{Proposition}
The identity \eqref{eqn:T-in-I} holds for all $n \in \bZ$.
\end{Proposition}

Next we present an analogue of \eqref{eqn:S-I} in 2D topological gravity:

\begin{Proposition} \label{prop:t-I}
The following identity holds:
\be
\sum_{n\in \mathbb{Z}}(-1)^nt_n\Gamma(-n+\frac{1}{2})z^{n-\frac{1}{2}}
= \sum_{p\in \mathbb{Z}} (-1)^pI_p\Gamma(-p+\frac{1}{2})\left(z-I_0\right)^{p-\frac{1}{2}},
\ee
where we use the following convention for expanding $\left(z-I_0\right)^{p-\frac{1}{2}}$:
\be
\left(z-I_0\right)^{p-\frac{1}{2}}
= z^{p-\frac{1}{2}} \sum_{k=0}^{\infty}\binom{p-\frac{1}{2}}{k}\left(\frac{-I_0}{z}\right)^k.
\ee
\end{Proposition}

\begin{proof}
By \eqref{eqn:T-in-I},
\begin{align*}
\sum_{n\in \mathbb{Z}}(-1)^nt_n\Gamma(-n+\frac{1}{2})z^{n-\frac{1}{2}}
=&
\sum_{n\in \mathbb{Z}}(-1)^n\sum_{k=0}^{\infty}\frac{(-1)^kI_0^k}{k!}I_{n+k}\Gamma(-n+\frac{1}{2})z^{n-\frac{1}{2}}\\
=&\sum_{p\in \mathbb{Z}} (-1)^pI_p\Gamma(-p+\frac{1}{2})
\sum_{k=0}^{\infty}\frac{\Gamma(-p+k+\frac{1}{2})}{k!\Gamma(-p+\frac{1}{2})}I_0^kz^{p-k-\frac{1}{2}}\\
=&\sum_{p\in \mathbb{Z}} (-1)^pI_p\Gamma(-p+\frac{1}{2})z^{p-\frac{1}{2}}
\sum_{k=0}^{\infty}\binom{p-\frac{1}{2}}{k}\left(\frac{-I_0}{z}\right)^k\\
=&\sum_{p\in \mathbb{Z}} (-1)^pI_p\Gamma(-p+\frac{1}{2})z^{p-\frac{1}{2}}
\left(1-\frac{I_0}{z}\right)^{p-\frac{1}{2}}\\
=&\sum_{p\in \mathbb{Z}} (-1)^pI_p\Gamma(-p+\frac{1}{2})\left(z-I_0\right)^{p-\frac{1}{2}}.
\end{align*}
\end{proof}

\subsection{$F_0^{2D}$ in $I$-coordinates with renormalized ghost variables}
When we impose the condition $t_{-n}=0$ for $n>0$,
the ghost variables $I_{-n}$ can be expressed as formal series in $\{t_{m}\}_{m\geq0}$ by the following equation:
\be
I_{-n}|_{t_{-m} = 0, m \geq 1}=\sum_{k=0}^{\infty}{t_{k}}\frac{I_0^{k+n}}{(k+n)!}.
\ee
The case of $I_{-1}$ is first introduced in \cite[(26)]{Zhou1}, and can be also expressed as:
\be\label{I-1}
I_{-1}|_{t_{-m} = 0, m \geq 1} = \sum_{k=0}^{\infty} I_k\frac{(-1)^k I_{0}^{k+1}}{(k+1)!}.
\ee
The following result generalizes this identity:

\begin{Proposition}
Under the same condition,
one can also rewrite $I_{-n}$ in terms of $\{I_{m}\}_{m\geq0}$ as follows:
\be
I_{-n}|_{t_{-m} = 0, m \geq 1}=\sum_{k=0}^{\infty}I_{k}\frac{(-1)^{k}I_{0}^{k+n}}{k!(n-1)!(k+n)}. \label{ghostIinI}
\ee
\end{Proposition}

\begin{proof}
By \eqref{eqn:T-in-I} again,
 we get:
\ben
I_{-n}|_{t_{-m} = 0, m \geq 1}&=&\sum_{m=0}^{\infty}{t_{m}}\frac{I_0^{m+n}}{(m+n)!} \\
&=&\sum_{m=0}^{\infty}\sum_{l=0}^{\infty}I_{m+l}\frac{(-1)^lI_{0}^l}{l!}\frac{I_0^{m+n}}{(m+n)!} \\
&=&\sum_{k=0}^{\infty}I_{k}I_{0}^{k+n}\sum_{m+l=k}\frac{(-1)^l}{l!(m+n)!} \\
&=&\sum_{k=0}^{\infty}I_{k}\frac{(-1)^{k}I_{0}^{k+n}}{k!(n-1)!(k+n)}.
\een
We have used the following identity:
\be
\sum_{m+l=k}\frac{(-1)^l}{l!(m+n+1)!}=\frac{(-1)^k}{n!k!(n+k+1)},\label{eqn:com}
\ee
This can be proved as follows. By the following well known identity:
\be
\binom{n+k+1}{l}=\binom{n+k}{l-1}+\binom{n+k}{l},
\ee
we have
\begin{align*}
\sum_{l=0}^{k}\binom{n+k+1}{l}(-1)^l=&\sum_{l=1}^{k}\binom{n+k}{l-1}(-1)^l+\sum_{l=0}^{k}\binom{n+k}{l}(-1)^l\\
=&\sum_{l=0}^{k-1}\binom{n+k}{l}(-1)^{l+1}+\sum_{l=0}^{k}\binom{n+k}{l}(-1)^l\\
=&\binom{n+k}{k}(-1)^k.
\end{align*}
Then, divide $(n+k+1)!$ on both side of this equation, we get \eqref{eqn:com}.
\end{proof}
We also denote the shifted $I$-coordinates by $\tilde{I}_n$, the definitions are as follows:
\be
\tilde{I}_{n}:=\sum_{k\geqslant0}\frac{I_0^k}{k!}\tilde{t}_{k+n}=I_{n}-\frac{I_0^{1-n}}{(1-n)!}.
\ee
where $\tilde{t}_k=t_k-\delta_{k,1}$.
Then one can see by above Proposition,
\be\label{eqn:shiftedghost}
\tilde{I}_{-n}|_{t_{-m} = 0, m \geq 1}=I_{-n}|_{t_{-m} = 0, m \geq 1}-\frac{I_0^{1+n}}{(1+n)!}
=\sum_{k=1}^{\infty}\tilde{I}_{k}\frac{(-1)^{k}I_{0}^{k+n}}{k!(n-1)!(k+n)}.
\ee
We can express the free energy of genus 0 with these notations as below:
\begin{Theorem}
The following formula for $F_0^{2D}$ holds:
\be \label{eqn:F02D-Ghost}
F^{2D}_0=\frac{1}{2}\sum_{n\geqslant1}(-1)^n\tilde{I}_{n}\tilde{I}_{-n-1}|_{t_{-m} = 0, m \geq 1}.
\ee
\end{Theorem}

\begin{proof}
Recall in Theorem \ref{IZ} we have proved the following expression of $F^{2D}_0$:
\be
F^{2D}_0=\frac{1}{6}I_0^3-\sum_{n\geqslant0}\frac{(-1)^nI_0^{n+2}}{(n+2)!} I_n +\frac{1}{2}\sum_{n,k\geqslant0}\frac{(-1)^{n+k}I_0^{n+k+1}}{k!n!(n+k+1)}I_n I_{k}.
\ee
One can rewrite this formula using $\tilde{I}_n$ as follows:
\be
F^{2D}_0=\frac{1}{2}\sum_{n,k\geqslant1}\frac{(-1)^{n+k}I_0^{n+k+1}}{k!n!(n+k+1)}\tilde{I}_n \tilde{I}_{k}.
\ee
Together with equation \eqref{eqn:shiftedghost}, the proof is completed.
\end{proof}

\begin{Remark}
Note we can rewrite \eqref{eqn:F01D} as follows:
\be\label{eqn:ghost1D}
F_0^{1D} = \tilde{I}_{-1}|_{t_{-n} = 0, n \geq 1}.
\ee
For the thin expansion of Hermitian one-matrix model of size $N \times N$,
because we have $F_0^N = N F_0^{1D}$,
we also have
\be\label{eqn:ghostN}
F_0^N = N \cdot  \tilde{I}_{-1}|_{t_{-n} = 0, n \geq 1}.
\ee
\end{Remark}

As a corollary,
we now recover the following formula in \cite[(5.12)]{IZ}:

\begin{Corollary} \label{cor:IZ}
The genus zero free energy of 2D gravity satisfies the following equality:
\begin{align}
F_0^{2D}&=\frac{I_0^3}{6}-\sum_{k\geqslant0}\frac{I_0^{k+2}}{k!(k+2)}t_{k}
+\frac{1}{2}\sum_{n\geqslant0}\sum_{k\geqslant0}\frac{I_0^{n+k+1}}{n!k!(n+k+1)}t_{n}t_{k}.
\end{align}
\end{Corollary}

\begin{proof}
We have
\begin{align*}
\sum_{p\geqslant1}(-1)^{p}\tilde{I}_p\tilde{I}_{-p-1}=&
\sum_{p\geqslant1}(-1)^{p}\sum_{i\geqslant0}\tilde{t}_{i+p}\frac{I_0^i}{i!}
\sum_{j\geqslant0}\tilde{t}_{j-p-1}\frac{I_0^j}{j!}\\
=&\sum_{n\in \mathbb{Z}}\tilde{t}_n\sum_{j-p-1=n}\sum_{i\geqslant0}
(-1)^p\tilde{t}_{i+p}\frac{I_0^{i+j}}{j!i!}\\
=&\sum_{n\geqslant0}\tilde{t}_{n}\sum_{p\geqslant1}\sum_{i\geqslant0}\tilde{t}_{i+p}
\frac{(-1)^pI_0^{n+p+i+1}}{i!(n+p+1)!}+\sum_{n<0}\tilde{t}_{n}\sum_{j\geqslant0}\sum_{i\geqslant0}
(-1)^{j-1-n}\tilde{t}_{i+j-1-n}\frac{I_0^{i+j}}{i!j!}\\
=&\sum_{n,k\geqslant0}\tilde{t}_{n}\tilde{t}_{k}\sum_{p+i=k}\frac{(-1)^pI_0^{n+k+1}}{i!(n+p+1)!}
+\sum_{n,k\geqslant0}\tilde{t}_{-n-1}\tilde{t}_{k+n}\sum_{i+j=k}(-1)^{j+n}\frac{I_0^{k}}{i!j!}\\
=&\sum_{n\geqslant0}\sum_{k\geqslant0}\frac{I_0^{n+k+1}}{n!k!(n+k+1)}\tilde{t}_{n}\tilde{t}_{k}+
\sum_{n\geqslant0}(-1)^n\tilde{t}_{n}\tilde{t}_{-n-1}.
\end{align*}
Hence,
\begin{align} \label{eqn:F0Ghost}
&\frac{1}{2}\sum_{n\geqslant1}(-1)^n\tilde{I}_{n}\tilde{I}_{-n-1}\\
=&\frac{1}{2}\sum_{n\geqslant0}\sum_{k\geqslant0}\frac{I_0^{n+k+1}}{n!k!(n+k+1)}
\tilde{t}_{n}\tilde{t}_{k}+\frac{1}{2}\sum_{n\geqslant0}(-1)^n(t_n-\delta_{n,1})t_{-n-1}\nonumber\\
=&\frac{I_0^3}{6}-\sum_{k\geqslant0}\frac{I_0^{k+2}}{k!(k+2)}t_{k}
+\frac{1}{2}\sum_{n\geqslant0}\sum_{k\geqslant0}\frac{I_0^{n+k+1}}{n!k!(n+k+1)}t_{n}t_{k}
+\frac{1}{2}\sum_{n\geqslant0}(-1)^n(t_n-\delta_{n,1})t_{-n-1}.\nonumber
\end{align}
The proof is completed by invoking  the condition that $t_{-n-1} = 0$ for $n \geq 0$
to \eqref{eqn:F02D-Ghost}.
\end{proof}

It is a remarkable coincidence that the right-hand side of \eqref{eqn:F0Ghost}
is almost the same as $\tilde{F}_0^{2D}$ in \eqref{eqn:Tilde-F02D}
used to interpret the special deformation of the Airy curve using ghost variables.
Inspired by this coincidence, we have the following result.

\begin{Theorem} \label{thm:IZinghost}
The extension $\tilde{F}_0^{2D}$ of the free energy $F_{0}^{2D}$ satisfies the following equality:
\begin{align}\label{eqn:tildeF0Ghost}
\tilde{F}_0^{2D}&=\frac{1}{6}I_0^3-\frac{1}{2}\sum_{n,k\geqslant0}\frac{(-1)^{n+k}I_0^{n+k+1}}{k!n!(n+k+1)}I_n I_{k}
+\sum_{n\geqslant1}(-1)^n(I_{n}-\delta_{n,1})I_{-n-1}.
\end{align}
\end{Theorem}

\begin{proof}
We have
\begin{align*}
\sum_{p\geqslant0}(-1)^{p}t_pt_{-p-1}=&
\sum_{p\geqslant0}(-1)^{p}\sum_{i\geqslant0}\frac{(-1)^iI_0^i}{i!}I_{i+p}
\sum_{j\geqslant0}\frac{(-1)^jI_0^j}{j!}I_{j-p-1}\\
=&\sum_{n\in \mathbb{Z}}\sum_{j-p-1=n}(-1)^p\frac{(-1)^jI_0^j}{j!}I_{j-p-1}
\sum_{i\geqslant0}\frac{(-1)^iI_0^i}{i!}I_{i+p}\\
=&\sum_{n\geqslant0}I_{n}\sum_{p\geqslant0}\sum_{i\geqslant0}\frac{(-1)^{n+i+1}I_0^{n+p+i+1}}{i!(n+p+1)!} I_{i+p}+
\sum_{n<0}I_{n}\sum_{j\geqslant0}\sum_{i\geqslant0}(-1)^{i-1-n}\frac{I_0^{i+j}}{i!j!}I_{i+j-1-n}\\
=&\sum_{n\geqslant0}I_{n}\sum_{k\geqslant0}I_{k}\sum_{p+i=k}\frac{(-1)^{p+n+k+1}I_0^{n+k+1}}{i!(n+p+1)!}+
\sum_{n\geqslant0}I_{-n-1}\sum_{k\geqslant0}I_{k+n}\sum_{i+j=k}(-1)^{i+n}\frac{I_0^{k}}{i!j!}\\
=&\sum_{n\geqslant0}\sum_{k\geqslant0}\frac{(-1)^{n+k+1}I_0^{n+k+1}}{n!k!(n+k+1)}I_{n}I_{k}+
\sum_{n\geqslant0}(-1)^nI_{n}I_{-n-1},
\end{align*}
and
\be
t_{-2}=\sum_{n\geqslant0}\frac{(-1)^nI_0^n}{n!}I_{n-2}.
\ee
By Theorem \ref{IZ},
\begin{align*}
\tilde{F}_0^{2D}=&\frac{1}{6}I_0^3-\sum_{n\geqslant0}\frac{(-1)^nI_0^{n+2}}{(n+2)!} I_n +\frac{1}{2}\sum_{n,k\geqslant0}\frac{(-1)^{n+k}I_0^{n+k+1}}{(n+k+1)k!n!}I_n I_{k}+t_{-2}+
\sum_{n\geqslant0}(-1)^nt_{n}t_{-n-1}\\
&=\frac{1}{6}I_0^3-\sum_{n\geqslant0}\frac{(-1)^nI_0^{n+2}}{(n+2)!} I_n +\frac{1}{2}\sum_{n,k\geqslant0}\frac{(-1)^{n+k}I_0^{n+k+1}}{(n+k+1)k!n!}I_n I_{k}\\
&+\sum_{n\geqslant0}\frac{(-1)^nI_0^n}{n!}I_{n-2}+
\sum_{n\geqslant0}\sum_{k\geqslant0}\frac{(-1)^{n+k+1}I_0^{n+k+1}}{n!k!(n+k+1)}I_{n}I_{k}+
\sum_{n\geqslant0}(-1)^nI_{n}I_{-n-1}\\
&=\frac{1}{6}I_0^3-\frac{1}{2}\sum_{n,k\geqslant0}\frac{(-1)^{n+k}I_0^{n+k+1}}{(n+k+1)k!n!}I_n I_{k}
+I_{-2}-I_0I_{-1}+\sum_{n\geqslant0}(-1)^nI_{n}I_{-n-1}\\
&=\frac{1}{6}I_0^3-\frac{1}{2}\sum_{n,k\geqslant0}\frac{(-1)^{n+k}I_0^{n+k+1}}{(n+k+1)k!n!}I_n I_{k}
+\sum_{n\geqslant1}(-1)^n(I_{n}-\delta_{n,1})I_{-n-1}.
\end{align*}
\end{proof}

\subsection{Special deformation of Airy curve in $I$-coordinates}
Recall the special deformation of Airy curve:
\be
w^{2D}=z^{\frac{1}{2}}-\sum_{n=0}^{\infty}\frac{t_n}{(2n-1)!!}z^{n-\frac{1}{2}} -\sum_{n=0}^{\infty}(2n+1)!!\frac{\partial F_0^{2D}}{\partial t_n}z^{-n-\frac{3}{2}}.
\ee
Now we consider a related deformation in a slightly different normalization as follows:
\be
y^{2D}:=z^{\frac{1}{2}}-\frac{1}{2\sqrt{\pi}}\sum_{n=0}^{\infty}(-1)^nt_n\Gamma(-n+\frac{1}{2})z^{n-\frac{1}{2}} -\frac{1}{2\sqrt{\pi}}\sum_{n=0}^{\infty}
\frac{\partial F_0^{2D}}{\partial t_n}\Gamma(n+\frac{3}{2})z^{-n-\frac{3}{2}}.
\ee
Now we rewrite it in $I$-coordinates, we get the following result:
\begin{Theorem}\label{Thm:y-I}In $I$-coordinates, we have£º
\begin{equation} \label{eqn:y-I}
y^{2D}=(z-I_0)^{\frac{1}{2}}-\frac{1}{2\sqrt{\pi}}
\sum_{n=1}^{\infty}(-1)^nI_n\Gamma(-n+\frac{1}{2})(z-I_0)^{n-\frac{1}{2}},
\end{equation}
and conversely, this property determined $F^{2D}_{0}$ uniquely.
\end{Theorem}

\begin{proof}
One can easily check that,
\be\label{eqn:Intk}
\frac{\pd \tilde{I}_n}{\pd t_k}=\frac{I_0^{k-n}}{(k-n)!}+\frac{I_0^{k}}{k!}\frac{\tilde{I}_{n+1}}{1-I_1},
\ee
for all $n,k\in \mathbb{Z}$. Therefore, by\eqref{eqn:F02D-Ghost}
\begin{align}\label{eqn:1ptF0}
\frac{\pd F^{2D}_0}{\pd t_k}=&\frac{1}{2}\sum_{n\geqslant1}(-1)^n
\bigg(\big(\frac{I_0^{k-n}}{(k-n)!}+\frac{I_0^{k}}{k!}\frac{\tilde{I}_{n+1}}{1-I_1}\big)\tilde{I}'_{-n-1}
+\tilde{I}_{n}\big(\frac{I_0^{k+n+1}}{(k+n+1)!}+\frac{I_0^{k}}{k!}\frac{\tilde{I}'_{-n}}{1-I_1}\big)\bigg)\\
=&\frac{1}{2}\sum_{n\geqslant0}(-1)^n\frac{I_0^{k-n}}{(k-n)!}\tilde{I}'_{-n-1}
+\frac{1}{2}\sum_{n\geqslant1}(-1)^n\tilde{I}_{n}\frac{I_0^{k+n+1}}{(k+n+1)!}\nonumber\\
=&\sum_{n\geqslant0}(-1)^n\frac{I_0^{k-n}}{(k-n)!}\tilde{I}'_{-n-1},\nonumber
\end{align}
where $\tilde{I}'_{n}=I_{n}|_{t_{-m}=0,\ m\geq1}$.
Hence,
\begin{align*}
\sum_{k=0}^{\infty}\frac{\partial F_0^{2D}}{\partial t_k}\Gamma(k+\frac{3}{2})z^{-k-\frac{3}{2}}
=&\sum_{k,n\geqslant0}(-1)^n\frac{I_0^{k-n}}{(k-n)!}\tilde{I}'_{-n-1}
\Gamma(k+\frac{3}{2})z^{-k-\frac{3}{2}}\\
=&\sum_{n\geqslant0}(-1)^n\tilde{I}'_{-n-1}\sum_{k\geqslant0}\frac{I_0^{k}}{k!}
\Gamma(k+n+\frac{3}{2})z^{-k-n-\frac{3}{2}}\\
=&\sum_{n\geqslant0}(-1)^n\tilde{I}'_{-n-1}\frac{\Gamma(n+\frac{3}{2})}{(z-I_0)^{n+\frac{3}{2}}}.
\end{align*}
Together with Proposition \ref{prop:t-I}, one has
\begin{align*}
y^{2D}=&z^{\frac{1}{2}}-\frac{1}{2\sqrt{\pi}}\sum_{n=0}^{\infty}(-1)^nt_n\Gamma(-n+\frac{1}{2})z^{n-\frac{1}{2}} -\frac{1}{2\sqrt{\pi}}\sum_{n=0}^{\infty}
\frac{\partial F_0^{2D}}{\partial t_n}\Gamma(n+\frac{3}{2})z^{-n-\frac{3}{2}}\\
=&z^{\frac{1}{2}}-\frac{1}{2\sqrt{\pi}}\sum_{n\in \mathbb{Z}}
(-1)^nI'_n\Gamma(-n+\frac{1}{2})\left(z-I_0\right)^{n-\frac{1}{2}}
-\frac{1}{2\sqrt{\pi}}\sum_{n\geqslant0}(-1)^n\tilde{I}'_{-n-1}
\frac{\Gamma(n+\frac{3}{2})}{(z-I_0)^{n+\frac{3}{2}}}\\
=&z^{\frac{1}{2}}-\frac{1}{2\sqrt{\pi}}\sum_{n\geqslant0}
(-1)^nI'_n\Gamma(-n+\frac{1}{2})\left(z-I_0\right)^{n-\frac{1}{2}}
+\frac{1}{2\sqrt{\pi}}\sum_{n\geqslant0}(-1)^n\frac{I_0^{n+2}}{(n+2)!}
\frac{\Gamma(n+\frac{3}{2})}{(z-I_0)^{n+\frac{3}{2}}}\\
=&(z-I_0)^{\frac{1}{2}}-\frac{1}{2\sqrt{\pi}}\sum_{n\geqslant1}
(-1)^nI_n\Gamma(-n+\frac{1}{2})\left(z-I_0\right)^{n-\frac{1}{2}}.
\end{align*}
The proof is completed.

\end{proof}
\begin{Remark}
When we impose the condition $t_{-n}=0$ for $n>0$, we have $F_0^{2D}=\tilde{F}_0^{2D}$. By Theorem \ref{thm:IZinghost}
\begin{align*}
F_0^{2D}&=\frac{1}{6}I_0^3-\frac{1}{2}\sum_{n,k\geqslant0}\frac{(-1)^{n+k}I_0^{n+k+1}}{(n+k+1)k!n!}I_n I_{k}
+\sum_{n\geqslant1}(-1)^n(I_{n}-\delta_{n,1})I_{-n-1}.
\end{align*}
if we take derivatives of $F_0^{2D}$ with respect to $\{I_{n}\}_{n\in \mathbb{Z}}$ formally, we have
\begin{align}
\frac{\partial F_0^{2D}}{\partial I_{n}}|_{t_{-m}=0,\ m\geq1}
=&-\sum_{k\geqslant0}\frac{(-1)^{n+k}I_0^{n+k+1}}{(n+k+1)k!n!} I_{k} +(-1)^{n}I_{-n-1}|_{t_{-m}=0,\ m\geq1}=0,\\
\frac{\partial F_0^{2D}}{\partial I_{-n-1}}=&(-1)^nI_n,
\end{align}
for $n\geqslant1$.
Consider the generating series
\be
\begin{split}
\sum_{n \in \bZ-\{0\}} \frac{\partial F_0^{2D}}{\partial I_{n}}|_{t_{-m}=0,\ m\geq1} \cdot x^{n+1}
= & \sum_{n \geq 1} (-1)^n (I_{n}-\delta_{n,1}) x^{-n}.
\end{split}
\ee
and its Laplace transform:
\be
\begin{split}
& \sum_{n \in \bZ-\{0\}} \frac{\partial F_0^{2D}}{\partial I_{n}} \cdot \int_0^{\infty}
\frac{1}{\sqrt{x}} e^{-\tilde{z}x}x^{n+1} dx
= \sum_{n \geq 1} (-1)^n(I_{n}-\delta_{n,1})\Gamma(-n+\frac{1}{2}) \tilde{z}^{n-1/2},
\end{split}
\ee
we get $y^{2D}$ by taking $\tilde{z}=z-I_0$.
\end{Remark}

\section{Applications: Constitutive Relations}
\label{sec:MFT}

In this Section we will use our results to rederive the following constitutive relations
in 2D topological gravity:
\be \label{eqn:DW-Const}
\corr{\sigma_i\sigma_j} = \frac{1}{i+j+1}u^{i+j+1}.
\ee 
This is (2.34) derived by Dijkgraaf and Witten \cite{DW} as an example
of constitutive relations  derived from topological recursion
relations in genus zero.
We will also derive the analogous relations for $F^{1D}_0$ and $F_0^N$.

\subsection{Constitutive relations in 2D topological gravity}

In our notations,
\eqref{eqn:DW-Const} is just:
\be
\frac{\pd^2 F^{2D}_0}{\pd t_j\pd t_k} = \frac{I_0^{k+j+1}}{k!j!(k+j+1)}.
\ee
This can be proved by our results as follows.
We have proved the following equation in the proof of Theorem \ref{Thm:y-I}:
\be
\frac{\pd F^{2D}_0}{\pd t_k}=\sum_{n\geqslant0}(-1)^n\frac{I_0^{k-n}}{(k-n)!}\tilde{I}'_{-n-1},
\ee
where $\tilde{I}'_{n}=I_{n}|_{t_{-m}=0,\ m\geq1}$.
By this equation and \eqref{eqn:Intk}, one has
\begin{align*}
\frac{\pd^2 F^{2D}_0}{\pd t_j\pd t_k}
=&\frac{\pd I_0}{\pd t_j}\sum_{n\geqslant0}(-1)^n\frac{I_0^{k-1-n}}{(k-1-n)!}\tilde{I}'_{-n-1}
+\sum_{n\geqslant0}(-1)^n\frac{I_0^{k-n}}{(k-n)!}\frac{\pd \tilde{I}'_{-n-1}}{\pd t_j}\\
=&\frac{1}{1-I_1}\frac{I_0^j}{j!}\sum_{n\geqslant0}(-1)^n\frac{I_0^{k-1-n}}{(k-1-n)!}\tilde{I}'_{-n-1}
+\sum_{n\geqslant0}(-1)^n\frac{I_0^{k-n}}{(k-n)!}
\big(\frac{I_0^{j+n+1}}{(j+n+1)!}+\frac{I_0^{j}}{j!}\frac{\tilde{I}'_{-n}}{1-I_1}\big)\\
=&\sum_{n\geqslant0}(-1)^n\frac{I_0^{k-n}}{(k-n)!}\frac{I_0^{j+n+1}}{(j+n+1)!}\\
=&\frac{I_0^{k+j+1}}{k!j!(k+j+1)}.
\end{align*}
This gives a proof of Constitutive Relations \eqref{eqn:DW-Const}.

We can now interpret the Itzykson-Zuber Ansatz 
as given by formulas \eqref{eqn:F02D}- \eqref{eqn:Fg2D}
from the point of view of constitutive relations. 
The formula for $F_0^{2D}$ can be obtained from the genus zero $n$-point functions on the small
phase space by changing $t_0^n$ to $(-1)^{n-1}I_0^n$, $t_k \to I_k$ for $k \geq 1$ as follows.
Because
\be
\corr{\tau_0^n}_0= \delta_{n,3},
\ee
the $0$-point function in genus zero on the small phase space is $\frac{t_0^3}{6}$,
it gives the term $\frac{I_0^3}{6}$.
The one-point function can be computed by
\be
\corr{\tau_n\tau_0^{m}}_0 = \delta_{m,n+2},
\ee
it is $\sum_{n \geq 0} \frac{t_0^{n+2}}{(n+2)!}t_n$,
and so it gives the term $- \sum_{n \geq 0} \frac{(-1)^nI_0^{n+2}}{(n+2)!} I_n$.
The two-point function is computed by
\be
\corr{\tau_j\tau_k\tau_0^{m}}_0
= \frac{(j+k)!}{j!k!} \delta_{m,j+k+1},
\ee
it is $\half\sum_{j,k\geq 0} \frac{t_0^{j+k+1}}{j!k!(j+k+1)}t_jt_k$,
and so it gives us the term
$\frac{1}{2} \sum_{j,k\geq 0} \frac{(-1)^{j+k}I_0^{j+k+1}}{j!k!(j+k+1)}I_jI_k$.
Since  for $g\geq 1$, $F_g$  does not involve $I_0$,
one can further restrict to the origin $t_0=0$ in the small phase space, compute
a few $n$-point functions and change $t_n$ to $I_n$ for $n \geq 1$.
Such formulas generalize the constitutive relations in the mean field theory considerations of
Dijkgraaf-Witten \cite{DW}.
So our discussions suggest that renormalization naturally leads to constitutive relations.

\subsection{Analogues for $F_0^{1D}$ and $F^N_0$}

By \eqref{eqn:ghost1D} and \eqref{eqn:Intk},one has:
\be
\frac{\pd F_0^{1D}}{\pd t_k}=\frac{I_0^{k+1}}{(k+1)!}.
\ee
Since $F^{N}_0=NF_0^{1D}$, one also has
\be
\frac{\pd F_0^{N}}{\pd t_k}=N\frac{I_0^{k+1}}{(k+1)!}.
\ee
We regard them as the analogues of \eqref{eqn:DW-Const}.
Furthermore, the Hessians are given by:
\be
\frac{\pd^2 F_0^{1D}}{\pd t_j\pd t_k}=\frac{I_0^{k}}{k!}\frac{\pd I_0}{\pd t_j}
=\frac{1}{1-I_1}\frac{I_0^{k+j}}{k!j!},
\ee
\be
\frac{\pd^2 F_0^{N}}{\pd t_j\pd t_k}=N\frac{I_0^{k}}{k!}\frac{\pd I_0}{\pd t_j}
=\frac{N}{1-I_1}\frac{I_0^{k+j}}{k!j!},
\ee
Note the appearance of $I_1$ on the right-hand sides means $F_0^{1D}$ and $F_0^N$ 
do not satisfy the topological recursion relations,
hence they do not give us topological field theories in the sense of \cite{DW}.

\section{Concluding Remarks}
\label{sec:Remarks}

In this paper, we have further studied the I-coordinates.
By rewriting $L_{0}^{2D}$ in I-coordinates, we have proved the Itzykson-Zuber ansatz.
Furthermore, we have developed the techniques of rewriting all the Virasoro constraints
for free energies in I-coordinates and solving free energies recursively.

As pointed out by the second named author in \cite{Zhou1},
we understand the I-variables as new coordinates on the big phase space.
In this paper, we have checked that at least in the cases of topological 1D gravity,
Hermitian one-matrix model and topological 2D gravity,
the free energies have good properties in these new coordinates.
We believe this is a general phenomenon
and hope to make generalizations in subsequent work.
We have also seen that the use of renormalized coupling constants shed some lights 
on the mean field theory approach to the original theories.

Furthermore,
we extend the definitions of $I_n$ for $n \geq -1$ to include $I_n$ for all $n\in \bZ$.
These are inspired by the introduction of $t_n$ for all $n\in \bZ$ in \cite{Zhou6}.
They suggest to study an even larger phase space to include the ghost variables.

A surprising consequence is that we discover a connection between
the emergent spectral curve of topological 2D gravity and the action of the topological 1D gravity.
The special deformations of  spectral curves of the three theories considered in this paper are:
\begin{align}
y^{1D}=&\frac{1}{\sqrt{2}}\sum_{n\geqslant0}\frac{t_n-\delta_{n,1}}{n!}z^n+\frac{\sqrt{2}}{z} +\sqrt{2}\sum_{n\geqslant1}\frac{n!}{z^{n+1}}\frac{\partial F_0^{1D}}{\partial t_{n-1}}.\\
y^{N}=&\frac{1}{\sqrt{2}}\sum_{n\geqslant0}\frac{t_n-\delta_{n,1}}{n!}z^n+\frac{\sqrt{2}N}{z} +\sqrt{2}\sum_{n\geqslant1}\frac{n!}{z^{n+1}}\frac{\partial F_0^{N}}{\partial t_{n-1}}.\\
y^{2D}=&z^{\frac{1}{2}}-\frac{1}{2\sqrt{\pi}}\sum_{n=0}^{\infty}(-1)^nt_n\Gamma(-n+\frac{1}{2})z^{n-\frac{1}{2}} -\frac{1}{2\sqrt{\pi}}\sum_{n=0}^{\infty}
\frac{\partial F_0^{2D}}{\partial t_n}\Gamma(n+\frac{3}{2})z^{-n-\frac{3}{2}}.
\end{align}
When rewritten in the $I$-coordinates they take the following unified form:
\begin{align}
y^{1D}=&\frac{\sqrt{2}}{z-I_0}+\frac{1}{\sqrt{2}}\sum_{n\geqslant1}\frac{I_n-\delta_{n,1}}{n!}(z-I_0)^n.\\
y^{N}=&\frac{\sqrt{2}N}{z-I_0}+\frac{1}{\sqrt{2}}\sum_{n\geqslant1}\frac{I_n-\delta_{n,1}}{n!}(z-I_0)^n.\\
y^{2D}=&\qquad \quad -\frac{\sqrt{\pi}}{2}
\sum_{n=1}^{\infty}\frac{I_n-\delta_{n,1}}{\Gamma(n+\frac{1}{2})}(z-I_0)^{n-\frac{1}{2}}.
\end{align}
Furthermore, if we define their action functions by:
\begin{align}
S^{1D}=&\int y^{1D} dz
=\sqrt{2}\log(z-I_0)+\frac{1}{\sqrt{2}}\sum_{n\geqslant1}\frac{I_n-\delta_{n,1}}{(n+1)!}(z-I_0)^{n+1},\\
S^{N}=&\int y^{N} dz
=\sqrt{2}N\log(z-I_0)+\frac{1}{\sqrt{2}}\sum_{n\geqslant1}\frac{I_n-\delta_{n,1}}{(n+1)!}(z-I_0)^{n+1},\\
S^{2D}=&\int y^{2D} dz
=\qquad \qquad \qquad \quad -\frac{\sqrt{\pi}}{2}
\sum_{n=1}^{\infty}\frac{I_n-\delta_{n,1}}{\Gamma(n+\frac{3}{2})}(z-I_0)^{n+\frac{1}{2}},
\end{align}
Then one has:
\begin{align}
\frac{\partial S^{1D}}{\partial t_0}=&-\frac{1}{1-I_1}\frac{\sqrt{2}}{z-I_0}+\frac{z-I_0}{\sqrt{2}},\\
\frac{\partial S^{N}}{\partial t_0}=&-\frac{1}{1-I_1}\frac{\sqrt{2}N}{z-I_0}+\frac{z-I_0}{\sqrt{2}},\\
\frac{\partial S^{2D}}{\partial t_0}=&\qquad \qquad \qquad \quad -(z-I_0)^{\frac{1}{2}}.
\end{align}
These are deformations of spectral curves:
\begin{align}
y^{1D}=&\frac{\sqrt{2}}{z}-\frac{z}{\sqrt{2}},\\
y^{N}=&\frac{\sqrt{2}N}{z}-\frac{z}{\sqrt{2}},\\
y^{2D}=& z^{\frac{1}{2}}.
\end{align}
I.e.,
the spectral curves are related to the action functions defined above in the following way:
\be
y^{*} = - \frac{\pd S^*}{\pd t_0} \biggl|_{I_0=I_1=0},
\ee
for $*=1D$, $N$, and $2D$.
We hope to understand and generalize this in future investigations.

\vspace{.2in}
{\bf Acknowledgements}.
The second named author is partly supported by NSFC grants 11661131005 and 11890662.

 \bibliographystyle{plain}

\end{document}